\documentclass[runningheads]{llncs}

\usepackage{amsmath,amssymb}
\usepackage{bussproofs}
\usepackage{mathtools}
\usepackage{stmaryrd}
\usepackage{multicol}
\usepackage{color}
\usepackage[normalem]{ulem}
\usepackage{tikz}
\usetikzlibrary{shapes,arrows}
\usepackage{bcprules}
\usepackage{afterpage,array,rotating}
\usepackage{ifthen}
\newboolean{printBibInSubfiles}
\setboolean{printBibInSubfiles}{true} 
\def\bib{\ifthenelse{\boolean{printBibInSubfiles}}
        {\bibliographystyle{splncs04}
         \bibliography{ref,abbrv,koba}}
        {}
    }

\newcommand\typical{\textsc{TyPiCal}}

\newcommand\seq[1]{\tilde{#1}}
\newcommand\zeroexp{\mathbf{0}}
\newcommand\outexp[3]{#1!(#2;#3).}
\newcommand\outatom[3]{#1!(#2;#3)}

\newcommand\soutatom[2]{#1!(#2)}
\newcommand\inexp[3]{#1?(#2;#3).}
\newcommand\inatom[3]{#1?(#2;#3)}
\newcommand\sinexp[2]{#1?(#2).}
\newcommand\sinatom[2]{#1?(#2)}
\newcommand\rinexp[3]{*#1?(#2;#3).}

\newcommand\srinexp[2]{*#1?(#2).}
\newcommand\srinatom[2]{*#1?(#2)}
\newcommand\letexp[3]{\mathbf{let}\ #1=#2\ \mathbf{in}\ #3}
\newcommand\ndletatom[1]{\mathbf{let}\ #1=\ndint \ \mathbf{in}}
\newcommand\ndletsatom[1]{\mathbf{let}\ \seq{#1}=\seq{\ndint} \ \mathbf{in}}
\newcommand\assexp[1]{\mathbf{Assume}(#1);}
\newcommand\assatom[1]{\mathbf{Assume}(#1)}
\newcommand\PAR{\mathbin{\,|\,}}
\newcommand\nuexp[1]{(\nu #1)}
\newcommand\COL{\mathbin{:}}
\newcommand\red{\to}
\newcommand\Chty[3]{\mathbf{ch}_{#1}(#2;#3)}
\newcommand\sChty[2]{\mathbf{ch}_{#1}(#2)}
\newcommand\mty{\sigma}
\newcommand\p{\vdash}
\newcommand\reg{\rho}
\newcommand\Def{\mathcal{D}}
\newcommand\Exp{E}
\newcommand\skipexp{(\,)}
\newcommand\nondet{\oplus}
\newcommand{\ASSUME}{\textbf{Assume}}

\newcommand\mrg\cup
\newcommand\sred{\seqto}
\newcommand{\nsred}{\dashrightarrow}
\newcommand\nsredv[1]{\nsred_{#1}}
\newcommand\Fname[1]{f_{#1}}
\newcommand\regof[1]{\rho_{#1}}
\newcommand\set[1]{\{#1\}}
\newcommand\prog[2]{(#1,#2)}
\newcommand\pre{\mathit{pred}}
\newcommand\fib{\mathit{fib}}
\newcommand\DEC{P_\text{dec}}
\newcommand\DECclosed{\DEC'}
\newcommand\TRUE{\mathbf{true}}
\newcommand\FALSE{\mathbf{false}}
\newcommand\rchepsilon[3]{\mathbf{ch}_{#1}(#2;#3)}
\newcommand\OK{\mathbf{ok}}
\newcommand\imp{\Longrightarrow}
\newcommand\hoice{\textsc{HoIce}}

\newcommand\inty[1]{#1_I}
\newcommand\outy[1]{#1_O}
\newcommand\subtype{<:}
\newcommand\ult{\textsc{Ultimate Automizer}}
\newcommand\zthree{\textsc{Z3}}
\newcommand\tyok[2]{#1\p #2\ \mathbf{ok}}
\newcommand\tyenvok[3]{#1; #2; #3 \p \mathbf{ok}}
\newcommand\FV{\mathbf{FV}}
\newcommand\dom{\mathit{dom}}

\newcommand\testcase[1]{``\textrm{#1}''}
\newcommand\Proj[2]{#1\mathbin{\downarrow_{#2}}}

\makeatletter
\newcommand{\len}[1]{\text{len}(#1)}
\newcommand\ty\iota
\newcommand\chty\kappa
\newcommand\env\Gamma
\newcommand\predenv\Phi
\newcommand\chenv\Delta
\newcommand{\sch}[3][\reg]{\Chty{#1}{#2}{#3}}
\newcommand\schan{\sch{\reg}{\seq{\ty}}{\seq{\chty}}}
\newcommand\srchan\schan
\newcommand\pred\phi
\newcommand\predvar{P}
\newcommand{\rch}[4]{\textbf{ch}_{#1}(#2; #3; #4)}
\newcommand{\ioch}[6]{\textbf{ch}_{#1}(#2; #3; #4; #5; #6)}
\newcommand{\fdef}[3]{#1(#2) = #3}
\newcommand{\ndint}{\star}
\newcommand{\ndletb@se}[3]{\textbf{let }#1 = #2 \textbf{ in } #3}
\newcommand{\ndletst@r}[2]{\ndletb@se{#1}{\ndint}{#2}}
\newcommand{\ndletnost@r}[2]{\ndletb@se{\seq{#1}}{\seq{\ndint}}{#2}}
\newcommand{\ndlet}{\@ifstar{\ndletst@r}{\ndletnost@r}}
\newcommand{\ifexp}[3]{\textbf{if}\ #1 \allowbreak \ \textbf{then}\ #2 \allowbreak \ \textbf{else}\ #3}
\newcommand{\op}{\mathit{op}}
\newcommand{\cname}[1]{\mathit{#1}}
\newcommand{\fname}[1]{\mathit{#1}}
\newcommand{\sequiv}{\equiv}
\newcommand\piequiv{\sequiv_{\pi}}
\newcommand{\expequiv}{\sequiv_{\text{E}}}

\newcommand\seqto{\rightsquigarrow}
\newcommand{\subdef}{\trianglelefteq}
\newif\if@draft
\@draftfalse
\newcommand{\sk}[1]
{
\if@draft%
{\small\textcolor{blue}{[#1 -sk]}}%
\else\ignorespaces%
\fi
}
\newcommand{\skchanged}[1]
{\if@draft\textcolor{blue}{#1}\else#1\fi}
\newcommand{\changed}[1]
{\if@draft\textcolor{red}{#1}\else#1\fi}
\newcommand{\nk}[1]
{
\if@draft%
{\small\textcolor{red}{[#1 -nk]}}%
\else\ignorespaces%
\fi
}
\newcommand{\sh}[1]
{
\if@draft%
{\small\textcolor{magenta}{[#1 -sh]}}%
\else\ignorespaces%
\fi
}
\newcommand{\shchanged}[1]
{\if@draft\textcolor{magenta}{#1}\else#1\fi}
\newcommand{\ry}[1]
{
\if@draft%
{\small\textcolor{green}{[#1 -ry]}}%
\else\ignorespaces%
\fi
}

\newif\if@aplas
\@aplasfalse
\newcommand{\ifaplas}[2]
{\if@aplas#1\else#2\fi}
\makeatother
\makeatletter
\RequirePackage[bookmarks,unicode,colorlinks=true]{hyperref}%
   \def\@citecolor{blue}%
   \def\@urlcolor{blue}%
   \def\@linkcolor{blue}%

\def\orcidID#1{\smash{\href{http://orcid.org/#1}{\protect\raisebox{-1.25pt}{\protect\includegraphics{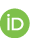}}}}}
\makeatother

\begin{document}
\setboolean{printBibInSubfiles}{false}
\title{Termination Analysis for the $\pi$-Calculus by Reduction to Sequential Program Termination}
\titlerunning{Termination Analysis for the \texorpdfstring{$\pi$}{Pi}-Calculus}
\author{
  Tsubasa Shoshi\inst{1}\orcidID{0000-0002-8164-0995} \and 
  Takuma Ishikawa\inst{1} \and
  Naoki Kobayashi\inst{1}\orcidID{0000-0002-0537-0604} \and
  Ken Sakayori\inst{1}\orcidID{0000-0003-3238-9279} \and
  Ryosuke Sato\inst{1}\orcidID{0000-0001-8679-2747} \and
  Takeshi Tsukada\inst{2}\orcidID{0000-0002-2824-8708}
}
\authorrunning{T. Shoshi et al.}
\institute{
  The University of Tokyo, Japan
  \and
  Chiba University, Japan
}
\maketitle              %
\begin{abstract}  %
  We propose an automated method for proving termination of \(\pi\)-calculus processes,
  based on a reduction to termination of sequential programs: 
  we translate a \(\pi\)-calculus process to a sequential program,
  so that the termination of the latter implies that of the former.
  We can then use an off-the-shelf termination
  verification tool to check termination of the sequential program.
  Our approach has been partially inspired by Deng and Sangiorgi's termination analysis for
  the \(\pi\)-calculus, and checks that there is no infinite chain of communications on
  replicated input channels, by converting such a chain of communications 
  to a chain of recursive function calls in the target sequential program.
  We have implemented an automated tool based on the proposed method and confirmed its effectiveness.

\end{abstract}

\section{Introduction}  \label{sec:introduction}

\newcommand\inexpIntro[3]{#1?(#2,#3).}
\newcommand\rinexpIntro[3]{*#1?(#2,#3).}
\newcommand\outexpIntro[3]{#1!(#2,#3).}
\newcommand\outatomIntro[3]{#1!(#2,#3)}

We propose a fully automated method for proving termination of \(\pi\)-calculus processes.
Although there have been a lot of studies on termination analysis for the \(\pi\)-calculus
and related calculi~\cite{Deng06IC,Demangeon07,SangiorgiTermination,KobayashiHybrid,Yoshida04IC,DBLP:journals/jlp/DemangeonHS10,Venet98SAS}, most of them have been rather theoretical,
and there have been surprisingly little efforts in developing  fully automated termination
verification methods and tools based on them. To our knowledge,
Kobayashi's \typical{}~\cite{TyPiCal,KobayashiHybrid} is the only exception that
can prove termination of \(\pi\)-calculus processes (extended with natural numbers)
fully automatically, but its termination analysis is quite limited (see Section~\ref{sec:relatedwork}).

Our method is based on a reduction to termination analysis for sequential programs:
we translate a \(\pi\)-calculus process \(P\) to a sequential program \(S_P\), so that
if \(S_P\) is terminating, so is \(P\). The reduction allows us to use
powerful, mature methods and tools
for termination analysis of sequential programs~\cite{heizmann2016ultimate,freqterm,DBLP:conf/lics/PodelskiR04,Kuwahara2014Termination,DBLP:journals/cacm/CookPR11}.

The idea of the translation is to convert a chain of communications on replicated input
channels to a chain of recursive function calls of the target sequential program.
Let us consider the following Fibonacci process:
\begin{align*}
    & \rinexpIntro{\fib}{n}{r}
        \ifexp{n<2}{ \soutatom{r}{1} \\ &\quad}
                   { \nuexp{s_1} \nuexp{s_2} (\outatomIntro{\fib}{n-1}{s_1} \PAR \outatomIntro{\fib}{n-2}{s_2} \PAR \sinexp{s_1}{x}\sinexp{s_2}{y}\soutatom{r}{x+y}) \\}
    & \PAR \outatomIntro{\fib}{m}{r}
\end{align*}
Here, the process
$\rinexpIntro{\fib}{n}{r} \ldots$ is a function server that computes the \(n\)-th Fibonacci number
in parallel and returns the result to \(r\),
and $\outatom{\fib}{m}{r}$ sends a request for computing the \(m\)-th Fibonacci number;
those who are not familiar with the syntax of the \(\pi\)-calculus may wish to consult
Section~\ref{sec:targetlanguage} first.
To prove that the process above is terminating for any integer \(m\),
it suffices to show that there is no infinite chain of communications on $\fib$:
\[
    \fib(m,r) \to \fib(m_1,r_1) \to \fib(m_2,r_2) \to \cdots.
\]
We convert the process above to the following program:\footnote{The actual translation
  given later is a little more complex.}
\begin{verbatim}
 let rec fib(n) = if n<2 then () else (fib(n-1) [] fib(n-2)) in
 fib(m)
\end{verbatim}
Here, \texttt{[]} represents the non-deterministic choice.
Note that, although the calculation of Fibonacci numbers is not preserved,
for each chain of communications on \texttt{fib}, there is a corresponding
sequence of recursive calls:
\[
\mathtt{fib}(m) \to \mathtt{fib}(m_1) \to \mathtt{fib}(m_2) \to \cdots.
\]
Thus, the termination of the sequential program above implies the termination of
the original process.
As shown in the example above, (i) each communication on a replicated input channel
is converted to a function call, (ii) each communication on a non-replicated input
channel is just removed (or, in the actual translation, replaced by a call of
a trivial function defined by \(f(\seq{x})=(\,)\)), and (iii) parallel composition
is replaced by a non-deterministic choice.
We formalize the translation outlined above and prove its correctness.

The basic translation sketched above sometimes loses too much information.
For example, consider the following process:
\begin{align*}
    & \rinexpIntro{\pre}{n}{r} \soutatom{r}{n-1} \\
    & \PAR \rinexpIntro{f}{n}{r} \ifexp{n<0}{ \soutatom{r}{1} }
                                       { \nuexp{s} (\outatomIntro{\pre}{n}{s} \PAR \sinexp{s}{x}\outatomIntro{f}{x}{r}) } \\
    & \PAR \outatomIntro{f}{m}{r}
\end{align*}
The translation sketched above would yield:
\begin{verbatim}
  let pred(n) = n-1 in
  let rec f(n) = if n<0 then () else (pred(n) [] f(*)) in
  f(m)
\end{verbatim}
Here, \texttt{*} represents a non-deterministic integer: since we have removed
the input $\sinatom{s}{x}$, we do not have information about the value of \( x \).
As a result, the sequential program above is non-terminating, although the original
process is terminating.
To remedy this problem, we also refine the basic translation above by using a refinement
type system for the \(\pi\)-calculus. Using the refinement type system,
we can infer that the value of \(x\) in the original process is less than \(n\),
so that we can refine the definition of \texttt{f} to:
\begin{verbatim}
 let rec f(n) = ... else (pred(n) [] let x=* in assume(x<n);f(x))
\end{verbatim}
The target program is now terminating, from which
we can deduce that the original process is also terminating.
We have implemented an automated tool based on the refined translation above.

The contributions of this paper are summarized as follows.
\begin{itemize}
\item The formalization of the basic translation from the \(\pi\)-calculus
  (extended with integers) to sequential programs, and a proof of its correctness.
\item The formalization of a refined translation based on a refinement type system.
\item An implementation of the refined translation, including automated refinement type
  inference based on CHC solving, and experiments to evaluate the effectiveness of
  our method.
\end{itemize}

The rest of this paper is structured as follows.
Section~\ref{sec:targetlanguage} introduces the source and target languages
of our translation.
Section~\ref{sec:approach} 
formalizes the basic translation, and proves its correctness.
Section~\ref{sec:refinement} refines the basic translation by using a refinement type system.
Section~\ref{sec:implementation} reports an implementation and experiments.
Section~\ref{sec:relatedwork} discusses related work,
and Section~\ref{sec:conclusion} concludes the paper.

\section{Source and Target Languages}  \label{sec:targetlanguage}

This section introduces the source and target languages for our reduction.
The source language is the
polyadic \(\pi\)-calculus~\cite{milner1993polyadic} extended with
integers and conditional expressions,
and the target language is a first-order functional language with non-determinism.

\subsection{$\pi$-Calculus}
\subsubsection{Syntax}

Below we assume
a countable set of variables ranged over by \(x, y, z, w,\!\ldots\)
and write \( \mathbb{Z} \) for the set of integers, ranged
over by \( i \).  We write $\tilde{\cdot}$ for (possibly empty) finite
sequences; for example, $\tilde{x}$ abbreviates a sequence
$x_1,\dots,x_n$.  We write \( \len{\tilde{x}} \) for the length of \(\seq{x}\) and
\(\epsilon\) for the empty sequence.

The sets of \emph{processes} and \emph{simple expressions},
ranged over by $P$ and $v$ respectively, are defined inductively
by: %
\begin{align*}
    &P \mbox{ (processes) }::= \ \zeroexp \mid \outexp{x}{\seq{v}}{\seq{w}}P \mid \inexp{x}{\seq{y}}{\seq{z}}P \mid \rinexp{x}{\seq{y}}{\seq{z}}P \mid (P_1 \PAR P_2) \mid \nuexp{x \COL \chty} P \\
    &\hphantom{P \mbox{ (processes) }::=}     \mid \ifexp{v}{P_1}{P_2} \mid \ndlet{x}{P} \\
    &v \mbox{ (simple expressions) }::= \ x \mid i \mid \op(\tilde{v})
\end{align*}
The syntax of processes on the first line is fairly standard, except that
the values sent along each channel consist of two parts: \(\seq{v}\) for integers,
and \(\seq{w}\) for channels; this is for the sake of technical convenience in
presenting the translation to sequential programs. The process \(\zeroexp\)
denotes an inaction, \(\outexp{x}{\tilde{v}}{\tilde{w}}P\) sends
a tuple \((\tilde{v},\tilde{w})\) along the channel \(x\) and behaves like \(P\),
and the process \(\inexp{x}{\tilde{y}}{\tilde{z}}P\) receives
a tuple \((\tilde{v},\tilde{w})\) along the channel \(x\), and behaves like
\([\seq{v}/\seq{y}, \seq{w}/\seq{z}]P\). We often just write \(\seq{v}\) for
\(\seq{v};\epsilon\) or \(\epsilon;\seq{v}\).
The process \(\rinexp{x}{\seq{y}}{\seq{z}}P\) represents infinitely many copies
of \(\inexp{x}{\seq{y}}{\seq{z}}P\) running in parallel.
The process \(P_1\PAR P_2\) runs \(P_1\) and \(P_2\) in parallel,
and \(\nuexp{x\COL\chty}{P}\) creates a fresh channel \shchanged{$x$} of type \(\chty\) (where types
will be introduced shortly) and behaves like \(P\).
The process \(\ifexp{v}{P_1}{P_2}\) executes \(P_1\) if the value of \(v \) is
non-zero, and \(P_2\) otherwise.
The process \(\ndlet{x}{P}\) instantiates the variables \(\seq{x}\) to
some integer values in a non-deterministic manner, and then behaves like \(P\).
The meta-variable \( \op \) ranges over integer operations such as \( + \) or \( \le \).

The free and bound variables are defined as usual.
The only binders are \(\nuexp{x\COL\chty}\)
(which binds \(x\)), \(\ndletsatom{x}\) (which binds \(\seq{x}\)),
\(\inexp{x}{\seq{y}}{\seq{z}}\)
and \(\rinexp{x}{\seq{y}}{\seq{z}}\) (which bind \(\seq{y}\) and \(\seq{z}\)).
Processes are identified up to renaming of bound variables,
and we implicitly apply \( \alpha \)-conversions as necessary.

We write \(P\red Q\) for the standard one-step reduction relation on processes.
The base cases of the communication are given by:
\begin{align*}
    \inexp{x}{\seq{y}}{\seq{z}}P_1 \PAR  \outexp{x}{\seq{v}}{\seq{w}}P_2 &\red [\seq{i}/ \seq{y}, \seq{w} / \seq{z} ]P_1 \PAR P_2 \\
    \rinexp{x}{\seq{y}}{\seq{z}}P_1 \PAR  \outexp{x}{\seq{v}}{\seq{w}}P_2 &\red \rinexp{x}{\seq{y}}{\seq{z}}P_1 \PAR [\seq{i}/ \seq{y}, \seq{w} / \seq{z}]P_1 \PAR P_2
\end{align*}
provided that \( \seq{v} \) evaluates to \( \seq{i} \).
\ifaplas{
  The full definition is given in
  the extended version~\cite{fullversion}}{The full definition is given in Appendix~\ref{sec:operational_semantics}}.
We say that a process \(P\) is \emph{terminating} if there is no infinite
reduction sequence \(P\red P_1\red P_2\red \cdots\).

In the rest of the paper, we consider only well-typed processes.
We write \(\ty\) for the type of integers.
The set of channel types, ranged over by \(\chty\), is given by:
\begin{align*}
  \chty ::= \Chty{\reg}{\tilde{\ty}}{\tilde{\chty}}
\end{align*}
The type $\Chty{\reg}{\tilde{\ty}}{\tilde{\chty}}$
describes channels used for transmitting a tuple \((\seq{v};\seq{w})\)
of integers \(\seq{v}\) and channels \(\seq{w}\) of types \(\seq{\chty}\).
Below we will just write \( \seq{\ty} \) for \( \seq{\ty}; \epsilon \) and \( \seq{\chty} \) for \( \epsilon ;\seq{\chty}\).
The subscript \(\reg\), called a \emph{region}, is a symbol that
abstracts channels; it is used in the translation to sequential programs.
For example, \(\Chty{\reg_1}{\ty}{\sChty{\reg_2}{\ty}}\) is the type of channels
that belong to the region \(\reg_1\) and are
used for transmitting a pair \((i,r)\) where
\(r\) is a channel of region \(\reg_2\) used for transmitting integers.
We use a meta-variable \(\mty\) for an integer or channel type.

Type judgments for processes and simple expressions are of the form \(\env;\chenv\p P\)
and \(\env;\chenv\p v:\mty\), where \(\env\) and \(\chenv\)
are sequences of bindings of the form
\(x\COL\ty\) and \(x\COL\chty\), respectively.
The typing rules are shown in Figure~\ref{fig:simple_type_system}.
Here \( \env; \chenv \vdash \seq{v} \COL \seq{\mty} \) means \( \env; \chenv \p v_i : \mty_i \) holds for each \( i \in \{ 1, \ldots, \len{\seq{v}} \} \).
We omit the explanation of the typing rules as they are standard.
\begin{figure}[tb]
    \centering
    \small
    \begin{minipage}{0.3\linewidth}
        \centering
        \begin{prooftree}
            \AxiomC{}
            \UnaryInfC{$\env; \chenv \vdash \zeroexp$}
        \end{prooftree}
    \end{minipage}
    \begin{minipage}{0.65\linewidth}
        \centering
        \begin{prooftree}
            \AxiomC{$\env; \chenv \vdash v \COL \ty$}
            \AxiomC{$\env; \chenv \vdash P_1$}
            \AxiomC{$\env; \chenv \vdash P_2$}
            \TrinaryInfC{$\env; \chenv \vdash \ifexp{v}{P_1}{P_2}$}
        \end{prooftree}
    \end{minipage}
    \\\vspace*{1ex}
    \begin{minipage}{0.38\linewidth}
        \centering
        \begin{prooftree}
            \AxiomC{$\env; \chenv \vdash P_1$}
            \AxiomC{$\env; \chenv \vdash P_2$}
            \BinaryInfC{$\env; \chenv \vdash P_1 \PAR P_2$}
        \end{prooftree}
    \end{minipage}
    \begin{minipage}{0.25\linewidth}
        \centering
        \begin{prooftree}
            \AxiomC{$\env; \chenv, x\COL\chty \vdash P$}
            \UnaryInfC{$\env; \chenv \vdash \nuexp{x \COL \chty}P$}
        \end{prooftree}
    \end{minipage}
    \begin{minipage}{.3\linewidth}
        \begin{prooftree}
            \AxiomC{$\env, \seq{x}\COL \seq{\ty}; \chenv \vdash P$}
            \UnaryInfC{$\env; \chenv \vdash \ndlet{x}{P}$}
        \end{prooftree}
    \end{minipage}
    \\\vspace*{1ex}
    \begin{minipage}{\linewidth}
        \centering
        \begin{prooftree}
            \AxiomC{$\env; \chenv \vdash x\COL \Chty{\reg}{\seq{\ty}}{\seq{\chty}}$}
            \AxiomC{$\env, \seq{y}\COL \seq{\ty}; \chenv, \tilde{z}\COL \tilde{\chty} \vdash P$}
            \BinaryInfC{$\env; \chenv \vdash \inexp{x}{\seq{y}}{\seq{z}}P $}
        \end{prooftree}
    \end{minipage}
    \\\vspace*{1ex}
    \begin{minipage}{\linewidth}
        \centering
        \begin{prooftree}
            \AxiomC{$\env; \chenv \vdash x\COL \Chty{\reg}{\seq{\ty}}{\seq{\chty}}$}
            \AxiomC{$\env; \chenv \vdash \seq{v}\COL \seq{\ty}$}
            \AxiomC{$\env; \chenv \vdash \seq{w}\COL \seq{\chty}$}
            \AxiomC{$\env; \chenv \vdash P$}
            \QuaternaryInfC{$\env; \chenv \vdash \outexp{x}{\seq{v}}{\seq{w}}P$}
        \end{prooftree}
    \end{minipage}
    \\\vspace*{1ex}
    \begin{minipage}{\linewidth}
        \begin{prooftree}
            \AxiomC{$\env; \chenv \vdash x\COL \Chty{\reg}{\seq{\ty}}{\seq{\chty}}$}
            \AxiomC{$\env, \seq{y}\COL \seq{\ty}; \chenv, \seq{z}\COL \seq{\chty} \vdash P$}
            \BinaryInfC{$\env; \chenv \vdash \rinexp{x}{\seq{y}}{\seq{z}}P$}
        \end{prooftree}
    \end{minipage}
    \\\vspace*{1ex}
    \begin{minipage}{0.2\linewidth}
        \centering
        \begin{prooftree}
            \AxiomC{$x\COL\ty \in \env$}
            \UnaryInfC{$\env; \chenv \vdash x\COL \ty$}
        \end{prooftree}
    \end{minipage}
    \begin{minipage}{0.2\linewidth}
        \centering
        \begin{prooftree}
            \AxiomC{$x\COL\chty \in \chenv$}
            \UnaryInfC{$\env; \chenv \vdash x\COL \chty$}
        \end{prooftree}
    \end{minipage}
    \begin{minipage}{0.18\linewidth}
        \centering
        \begin{prooftree}
            \AxiomC{}
            \UnaryInfC{$\env; \chenv \vdash i\COL \ty$}
        \end{prooftree}
    \end{minipage}
    \begin{minipage}{0.25\linewidth}
        \centering
        \begin{prooftree}
            \AxiomC{$\env; \chenv \vdash \tilde{v}\COL \seq{\ty}$}
            \UnaryInfC{$\env; \chenv \vdash \op(\seq{v})\COL \ty$}
        \end{prooftree}
    \end{minipage}
    \normalsize
    \caption{The typing rules of the simple type system for the $\pi$-calculus}
    \label{fig:simple_type_system}
\end{figure}

\subsection{Sequential Language}
We define the target language of our translation, which
is a first-order functional language with non-determinism.

A \emph{program} is a pair \((\Def, \Exp)\) consisting of (a set of)
\emph{function definitions} \(\Def\) and
an \emph{expression} \(\Exp\), defined by:
\begin{align*}
    \mathcal{D} \ \text{(function definitions)}        \, ::=& \ \{ \fdef{f_1}{\tilde{x}_1}{E_1}, \ldots, \fdef{f_n}{\tilde{x}_n}{E_n} \} \\
    E           \ \text{(expression)}                  \, ::=& \ \skipexp
    \mid \ndlet{x}{E} \mid f(\tilde{v}) \mid E_1 \nondet E_2 \\
                                                             & \ \mid \textbf{if } v \textbf{ then } E_1 \textbf{ else } E_2 \mid \textbf{Assume}(v); E \\
    v           \ \text{(simple expressions)}                       \, ::=& \ x \mid i \mid \op(\tilde{v})
\end{align*}
In a function definition
\( \fdef{f_i}{x_1, \ldots , x_{k_i}}{\Exp_i} \),
the variables \( x_1, \ldots, x_{k_i} \)
are bound in \( \Exp_i \); we identify function definitions up to renaming of bound
variables, and implicitly apply \(\alpha\)-conversions.
The function names \( f_1, \ldots, f_n \) need not be distinct from each other.
  If there are more than one definition for \( f \), then one of the definitions
  will be non-deterministically used when \( f \) is called.
We explain the informal meanings of the nonstandard expressions.
The expression \( \ndlet{x}{E} \)
instantiates \(\seq{x}\) to some integers in a non-deterministic manner.
The expression \(E_1\nondet E_2\) non-deterministically evaluates to
\(E_1\) or \(E_2\).
The expression
\(\textbf{Assume}(v); E \) evaluates to \(E\) if \(v\) is non-zero;
otherwise the whole program is aborted.
The other expressions are standard and their meanings should be clear.

We write \((\Def,\Exp)\sred (\Def,\Exp')\)
for the one-step reduction relation, whose definition is given
in \ifaplas{the extended version~\cite{fullversion}}{Appendix~\ref{sec:operational_semantics}}.
We say that a program is terminating if there is no infinite
reduction sequence.

\section{Basic Transformation}  \label{sec:approach}

This section presents our transformation from
a  \(\pi\)-calculus process to a sequential program,
so that
if the transformed program is terminating then the original process is terminating.

As explained in Section~\ref{sec:introduction}, the idea is to transform
an infinite chain of message passing on replicated input channels to
an infinite chain of recursive function calls.
Table~\ref{tab:trans} summarizes the correspondence between
processes and  sequential programs.
As shown in the table, a replicated input process is transformed to
a function definition, whereas a non-replicated input process is
just ignored, and integer bound variables are non-deterministically instantiated.
Note that channel arguments \(\seq{z}\) are ignored in both cases.
Instead, we prepare a global function name \(\Fname{\reg}\) for each
region \(\reg\); \(\regof{x}\) in the table indicates the region assigned to
the channel type of \(x\).\footnote{Thus, 
  the simple type system with
  ``regions'' introduced in the previous section is used here as a simple
  may-alias analysis.
  If \(x\) and \(y\) may be bound to the same channel during reductions,
  the type system assigns the same region to \(x\) and \(y\),
  hence \(x\) and \(y\) are
  mapped to the same function name \(\Fname{\regof{x}}\) by our transformation.}

\begin{table}[tbp]
  \caption{Correspondence between processes and sequential programs}
  \label{tab:trans}
\begin{tabular}{|l|l|}
  \hline
  processes & sequential programs \\
  \hline
  \hline
  replicated input (\(\rinexp{x}{\seq{y}}{\seq{z}}\cdots\)) & function definition  \(\Fname{\regof{x}}(\seq{y})=\cdots\)\\
  \hline
  non-replicated input (\(\inexp{x}{\seq{y}}{\seq{z}}\cdots\)) &
  non-deterministic instantiation (\(\ndlet{y}\cdots\))\\
  \hline
  output (\(\outexp{x}{\seq{v}}{\seq{w}}\cdots\)) &
    function call (\(\Fname{\regof{x}}(\seq{v})\nondet \cdots\))\\
    \hline
    parallel composition (\(\cdots\PAR\cdots\)) &
    non-deterministic choice (\(\cdots\nondet \cdots\))\\
    \hline
\end{tabular}
\end{table}

We define the transformation relation
$\env; \chenv \vdash P \Rightarrow \prog{\Def}{\Exp}$,
which means that the \(\pi\)-calculus process \(P\) well-typed
under \(\env;\chenv\) is transformed to the sequential program \((\Def,\Exp)\).
The relation is defined by the rules in Figure~\ref{fig:program_transformation}.
\begin{figure}[tb]
  \typicallabel{SX-RIn}
  \infrule[SX-Nil]{}
          {\env; \chenv \vdash \textbf{0} \Rightarrow
            \prog{\set{ \fdef{\Fname{\reg}}{\seq{y}}{\skipexp}
              \mid x\COL\Chty{\reg}{\seq{\ty}}{\seq{\chty}} \in \chenv, 
                   \len{\seq{y}} = \len{\seq{\ty}} }}{\skipexp}}
          \vspace*{1ex}
          \infrule[SX-In]
            {\env; \chenv \vdash x : \Chty{\reg}{\seq{\ty}}{\seq{\chty}}\andalso
            \env, \seq{y} : \seq{\ty}; \chenv, \seq{z} : \seq{\chty} \vdash P \Rightarrow \prog{\Def}{\Exp}}
            {\env; \chenv \vdash \inexp{x}{\seq{y}}{\seq{z}}P \Rightarrow
              \prog{\ndlet{y}{\Def}}{\ndlet{y}{\Exp}}}
          \vspace*{1ex}
\infrule[SX-RIn]
{\env; \chenv \vdash x : \Chty{\reg}{\seq{\ty}}{\seq{\chty}}\andalso
  \env, \seq{y} : \seq{\ty}; \chenv, \seq{z} : \seq{\chty} \vdash P \Rightarrow
  \prog{\Def}{\Exp}}
{\env; \chenv \vdash \rinexp{x}{\seq{y}}{\seq{z}}P \Rightarrow
  \prog{\set{ \fdef{\Fname{\reg}}{\seq{y}}{\Exp} } \mrg (\ndlet{y}{\Def})} {\skipexp}}
\vspace*{1ex}
\infrule[SX-Out]
{\env; \chenv \vdash x : \Chty{\reg}{\seq{\ty}}{ \seq{\chty}}\andalso
\env; \chenv \vdash \seq{v} : \seq{\ty}\andalso
\env; \chenv \vdash \seq{w} : \seq{\chty}\andalso
\env; \chenv \vdash P \Rightarrow \prog{\Def}{\Exp}}
{\env; \chenv \vdash \outexp{x}{\seq{v}}{\seq{w}}P \Rightarrow
  \prog{\mathcal{D}}{\Fname{\reg}(\seq{v}) \nondet \Exp}}

\infrule[SX-Par]
{\env; \chenv \vdash P_1 \Rightarrow \prog{\Def_1}{\Exp_1}\andalso
\env; \chenv \vdash P_2 \Rightarrow \prog{\Def_2}{\Exp_2}}
{\env; \chenv \vdash P_1 \mid P_2 \Rightarrow \prog{\Def_1 \mrg \Def_2}{\Exp_1 \nondet \Exp_2}}

\vspace*{1ex}
\infrule[SX-Nu]
        {\env; \chenv, x : \chty \vdash P \Rightarrow \prog{\Def}{\Exp}}
        {\env; \chenv \vdash \nuexp{x \COL \chty}P \Rightarrow \prog{\Def}{\Exp}}
          \vspace*{1ex}

\infrule[SX-If]
{\env; \chenv \vdash v : \ty\andalso
 \env; \chenv \vdash P_1 \Rightarrow \prog{\Def_1}{\Exp_1}\andalso
 \env; \chenv \vdash P_2 \Rightarrow \prog{\Def_2}{\Exp_2}}
{\env; \chenv \vdash \ifexp{v}{P_1}{P_2} \Rightarrow
  \prog{\Def_1 \mrg \Def_2}{\ifexp{v}{\Exp_1}{\Exp_2}}}

          \vspace*{1ex}

\infrule[SX-LetND]
{\env, \seq{x}: \seq{\ty}; \chenv \vdash P \Rightarrow \prog{\Def}{\Exp}}
{\env; \chenv \vdash \ndlet{x}{P} \Rightarrow
  \prog{\ndlet{x}{\Def}}{\ndlet{x}{\Exp}}}

    \begin{align*}
        \ndlet{x}{\Def} \coloneqq
& \{ \fdef{f}{\tilde{y}}{(\ndlet{x}{\Exp})} \mid \fdef{f}{\tilde{y}}{E} \in \Def \} 
    \end{align*}
    \normalsize
    \caption{The rules of simple type-based program transformation}
    \label{fig:program_transformation}
\end{figure}

We explain some key rules.
In \rn{SX-Nil}, \(\zeroexp\) is translated to \((\Def,\skipexp)\),
where \(\Def\) is the set of trivial function definitions.
In \rn{SX-In}, a (non-replicated) input is just removed, 
and the bound variables are instantiated to non-deterministic integers;
this is because we have no information about \(\seq{y}\); this will be refined
in Section~\ref{sec:refinement}. In contrast,
in \rn{SX-RIn}, a replicated input is converted to a function definition.
Since \(\Def\) generated from \(P\) may contain \(\seq{y}\), they are
 bound to non-deterministic integers and merged with the new definition for \(\Fname{\reg}\).
In \rn{SX-Out}, an output is replaced by a function call.
In \rn{SX-Par}, parallel composition is replaced by non-deterministic choice.

\begin{example}
\label{ex:fib}
\newcommand{\FIB}{P_\text{fib}}
    Let us revisit the Fibonacci example used in the introduction to explain the actual translation.
    Using the syntax we introduced, the Fibonacci process \( \FIB \) can now be defined as:
    \begin{align*}
        & \nuexp{\fib\COL \Chty{\reg_1}{\ty}{\sChty{\reg_2}{\ty}}}\rinexp{\fib}{n}{r} \\
        & \quad \ifexp{n<2}{\soutatom{r}{1}}{(\nu r_1\COL \sChty{\reg_2}{\ty})(\nu r_2\COL \sChty{\reg_2}{\ty})\\
        & \quad (\outatom{\fib}{n-1}{r_1} \PAR \outatom{\fib}{n-2}{r_2} \PAR \sinexp{r_1}{x} \sinexp{r_2}{y} \soutatom{r}{x+y})} \\
        & \PAR \letexp{m}{\ndint}{(\nu r\COL \sChty{\reg_2}{\ty}) \outatom{\fib}{m}{r}}
    \end{align*}
    Note that \( \nuexp{\fib} \) and \( \ndletatom{m} \) have been added to close the process.
    We can derive $\emptyset; \emptyset \vdash \FIB \Rightarrow \prog{\Def}{\Exp}$,
    where $\Def$ and $\Exp$ are given as follows:\footnotemark
    \footnotetext{The program written here has been simplified for the sake of readability.
    For instance, we removed some redundant \( \skipexp \), trivial function definitions, and unused non-deterministic integers.
    The other examples that will appear in this paper are also simplified in the same way.
    }
\begin{align*}
        \Def &=  \{ \fdef{\Fname{\reg_1}}{z}{\ifexp{z<2}{\Fname{\reg_2}(1)}{( \Fname{\reg_1}(z-1) \nondet \Fname{\reg_1}(z-2) \\ 
               &\qquad\qquad\qquad \nondet \letexp{x}{\ndint}{\letexp{y}{\ndint}{\Fname{\reg_2}(x+y)}} )}}, \\
               &\ \quad  \fdef{\Fname{\reg_2}}{z}{\skipexp} \} \\ 
        \Exp &= \letexp{m}{\ndint}{\Fname{\reg_1}(m)}
    \end{align*}
    Here \( \Fname{\reg_1} \) is the ``Fibonacci function'' because \( \reg_1 \) is the region assigned to the channel \( \cname{fib}\) in \( \FIB \).
    The function call \( \Fname{\reg_2} (x + y)\) corresponds to the output \( \soutatom{r}{x + y} \); the argument of the function call is actually a nondeterministic integer because \( \sinatom{r}{x} \) and \( \sinatom{r}{y} \) are translated to non-deterministic instantiations.
    Since the program $(\Def, \Exp)$ is terminating,
    we can verify that $\FIB$ is also terminating.
    \qed
\end{example}

\begin{example}
\label{ex:nested_rep}
To help readers understand the rule \rn{SX-RIn}, we consider the following process, which contains a nested input:
\begin{align*}
 \rinexp{f}{x}{r}\srinexp{g}{y, z}(\ifexp{y \leq 0}{\soutatom{r}{z}}{\soutatom{g}{y - 1, x + z}} ) \PAR \outexp{f}{2}{r}\soutatom{g}{3,0}
\end{align*}
where \( f \COL \Chty{\reg_1}{\ty}{\sChty{\reg_2}{\ty}} \) and \( g \COL \sChty{\reg_3}{\ty, \ty} \).
This process computes \( x * y + z \) (which is \( 6 \) in this case) and returns that value using \( r \).
This program is translated to:
\begin{align*}
  &\fdef{\Fname{\reg_1}}{x}{\skipexp}\qquad \fdef{\Fname{\reg_2}}{z}{\skipexp} \\
  &\fdef{\Fname{\reg_3}}{y, z}{\ndlet*{x}{\ifexp{y \leq 0}{\Fname{\reg_2}(z)}{} } \Fname{\reg_3}(y - 1, x + z)}
\end{align*}
with the main expression  \( \Fname{\reg_1}(2) \nondet \Fname{\reg_3}(3, 0) \).
Note that the body of \( \Fname{\reg_1}\), which is the function corresponding to \( f \), is \( \skipexp \).
This is because when the rule \rn{SX-RIN} is applied to 
\( \srinatom{g}{y, z}\ldots\),
the main expression of the translated program becomes \( \skipexp \).
Observe that the function definition for \( \Fname{\reg_3} \) still contains a free variable \( x \) at this moment.
Then \( \Fname{\reg_3} \) is closed by \( \ndletatom{x} \) when we apply the rule \rn{SX-RIn} to \( \srinatom{f}{x; r}\ldots \).
We can check that the above program is terminating, and thus we can verify that the original process is terminating.
Note that some precision is lost in the application of \rn{SX-RIn} above
since we cannot track the relation between the argument of \( \Fname{\reg_1} \)
and the value of \( x \) used inside \( \Fname{\reg_3}\). This loss causes a problem if, for example, the condition \(y\le 0\) in
  the process above is replaced with \(y\le x\). The body of
  \(\Fname{\reg_3}\) would then become
  \(\ndlet*{x}{\textbf{if } y\le x\ \cdots}\), hence the sequential program would be
    non-terminating.
 \qed
\end{example}

\begin{remark}
A reader may wonder why
a non-replicated input is removed in \rn{SX-In},
rather than translated to a function definition as done for a replicated input.
It is actually possible to obtain a sound transformation even if we treat
non-replicated inputs in the same manner as replicated inputs,
but we expect that our approach of removing non-replicated inputs often works
better.
For example,
  consider $\sinexp{x}{y}\soutatom{x}{y} \PAR \soutatom{x}{0}$.
  Our translation generates
  $\prog{\{\fdef{\Fname{\reg_x}}{z}{\skipexp}\}}{(\ndlet*{y}{\Fname{\reg_x}(y)}) \nondet \Fname{\reg_x}(0)}$
  which is terminating,
  whereas if we treat the input in the same way as a replicated input,
  we would obtain
  $\prog{\{\fdef{\Fname{\reg_x}}{z}{\Fname{\reg_x}(z)}\}}{\Fname{\reg_x}(0)}$
  which is not terminating.
 Our approach also has some defect.
  For example, consider
  $\soutatom{x}{0} \PAR \sinexp{x}{y} \ifexp{y=0}{\zeroexp}{\Omega}$
  where $\Omega$ is a diverging process. 
  Our translation yields
  $\prog{\{\fdef{\Fname{\reg_x}}{z}{\skipexp}\}}{\Fname{\reg_x}(0) \nondet \ndlet*{y}{\ifexp{y=0}{\skipexp}{\Omega'}}}$
  which is non-terminating.
  On the other hand, if we treat the input like a replicated input,
  we would obtain
  $\prog{\{\fdef{\Fname{\reg_x}}{z}{\ifexp{z=0}{\skipexp}{\Omega'}}\}}{\Fname{\reg_x}(0)}$
  which is terminating.
  This issue can, however,
  be mitigated by the extension with refinement types in Section~\ref{sec:refinement}.
Our choice of removing non-replicated inputs is also 
consistent with Deng and Sangiorgi's type system~\cite{Deng06IC}, which
prevents an infinite chain of communications on replicated input channels by using types
and ignores non-replicated inputs.
\qed
\end{remark}

The following theorem states the soundness of our transformation.
\begin{theorem}[soundness]  \label{thm:soundness}
  Suppose  $\emptyset; \emptyset \vdash P \Rightarrow \prog{\Def}{\Exp}$.
    If $(\Def, \Exp)$ is terminating, then so is $P$.
\end{theorem}
We briefly explain the proof strategy; see \ifaplas{the extended version~\cite{fullversion}}{Appendix~\ref{sec:soundness}}
for the actual proof.
Basically, our idea is to show that the translated program simulates the original process.
Then we can conclude that if the original process is non-terminating then so is the sequential program.
However, there is a slight mismatch between the reduction of a process and that of the sequential program that we need to overcome.
Recall that \( \srinexp{f}{x}P \PAR \soutatom{f}{1} \PAR \soutatom{f}{2} \) is translated to \( \Fname{\reg_f}(1) \nondet \Fname{\reg_f}(2) \) with a function definition for \( \Fname{\reg_f} \).
In the sequential program, we need to make a ``choice'', e.g.~if \( \Fname{\reg_f}(1) \) is called, we cannot call \( \Fname{\reg_f}(2) \) anymore.
On the other hand, the output \( \soutatom{f}{2} \) can be used even if \( \soutatom{f}{1} \) has been used before.
To fill this gap, we introduce a non-standard reduction relation, which does not discard branches of non-deterministic choices and show the simulation relation using that non-standard semantics.
Then we show that if there is an infinite non-standard reduction sequence, then there is an infinite subsequence that corresponds to a reduction along a certain choice of non-deterministic branches.
This step is essentially a corollary of the K\"onig's Lemma.
This is because the infinite non-standard reduction sequence can be reformulated as an infinite tree in which branches correspond to non-deterministic choices \( \oplus \) (thus the tree is finitely branching) and paths correspond to reduction sequences.

The following example indicates that the basic transformation is
sometimes too conservative.
\begin{example}
  \label{ex:weakeness-of-basic-transformation}
Let us consider the following process \( \DEC \):
\begin{align*}
    &\rinexp{\pre}{n}{r} \soutatom{r}{n-1} \\
    & \PAR \rinexp{f}{n}{r} \ifexp{n<0}{ \soutatom{r}{1} }{ \nuexp{s \COL \sChty{\reg_2}{\ty} } (\outatom{\pre}{n}{s} \PAR \sinexp{s}{x}\outatom{f}{x}{r}) } \\
    &\PAR \outatom{f}{m}{r}
\end{align*}
where
\( \pre\COL \Chty{\reg_1}{\ty}{\sChty{\reg_2}{\ty}} \), \(\cname{f} \COL \Chty{\reg_3}{\ty}{\sChty{\reg_4}{\ty}} \) and \( r \COL \sChty{\reg_4}{\ty}\).
This process, which also appeared in the introduction, keeps on decrementing the integer \( m \) until it gets negative and then returns \( 1 \) via \( r \).
We can turn this process into a closed process \( \DECclosed \) by restricting the names \( \pre \), \( \cname{f} \), \( \cname{r} \) and adding \( \ndletatom{m} \) in front of the process.
Note that \( \DECclosed \) is terminating.

The process \( \DECclosed \) is translated to:
\begin{align*}
  &\fdef{\Fname{\reg_1}}{n}{\Fname{\reg_2}(n-1)}, 
  \qquad \fdef{\Fname{\reg_2}}{x}{\skipexp},
  \qquad \fdef{\Fname{\reg_4}}{x}{\skipexp}, \\
  & \fdef{\Fname{\reg_3}}{n}{\ifexp{n<0}{\Fname{\reg_4}(1) \\ &\qquad\qquad }
            {(\Fname{\reg_1}(n) \nondet \letexp{x}{\ndint}{\Fname{\reg_3}(x)}})}
\end{align*}
with the main expression \( \ndlet*{m}{\Fname{\reg_3}(m)} \).
Observe that the function \( f_{\reg_3}\) is applied to a non-deterministic integer, not \( n - 1\).
Thus, this program is not terminating, meaning that we fail to verify that the original process is terminating.
This is due to the shortcoming of our transformation that all the integer   values received by non-replicated inputs are replaced by non-deterministic integers.
This problem is addressed in the next section.
\qed
\end{example}

\section{Improving Transformation Using Refinement Types}  \label{sec:refinement}

In this section, we refine the basic transformation in the previous section
by using a refinement type system.

Recall that in Example~\ref{ex:weakeness-of-basic-transformation},
the problem was that information about values received by non-replicated
inputs was completely lost.
By using a refinement type system for the \(\pi\)-calculus,
we can statically infer that \(x<n\) holds between \(x\) and \(n\) in
the process in Example~\ref{ex:weakeness-of-basic-transformation}.
Using that information, we can transform the process in Example~\ref{ex:weakeness-of-basic-transformation} and obtain
\[
    \fdef{\Fname{\reg_3}}{n}{\ifexp{n<0}{\cdots}
            {(\Fname{\reg_1}(n) \nondet \letexp{x}{\ndint}{\assexp{x < n}\Fname{\reg_3}(x)}})}
\]
for the definition of \( \Fname{\reg_3}\).
This is sufficient to conclude that the resulting program is terminating.

In the rest of this section, we first introduce a refinement type system
in Section~\ref{sec:rtype} and explain the refined transformation in Section~\ref{sec:rx}.
We then discuss how to automatically infer refinement types and
achieve the refined transformation in Section~\ref{sec:inference}.

\subsection{Refinement Type System}
\label{sec:rtype}

The set of \emph{refinement channel types}, ranged over by $\chty$, is given by:
\begin{align*}
    \chty ::= \rch{\reg}{\seq{x}}{\pred}{\seq{\chty}}
\end{align*}
Here, \(\pred\) is a formula of integer arithmetic.
We sometimes write just
\(\rchepsilon{\reg}{\seq{x}}{\pred}\) for
\(\rch{\reg}{\seq{x}}{\pred}{\epsilon}\).
Intuitively, 
$\rch{\reg}{\tilde{x}}{\pred}{\tilde{\chty}}$
describes channels that are used for transmitting
a tuple \((\seq{x};\seq{y})\) such that (i) \(\seq{x}\) are integers
that satisfy \(\pred\), and (ii) \(\seq{y}\) are channels of types \(\seq{\chty}\).
For example, the type
\(\rch{\reg_1}{x}{\TRUE}{\rchepsilon{\reg_2}{z}{z<x}}\)
describes channels used for transmitting a pair
\((x, y)\), where \(x\) may be any integer, and \(y\) must be a channel of
type \(\rchepsilon{\reg_2}{z}{z<x}\), i.e.,
a channel used for passing an integer \(z\) smaller than \(x\).%
Thus, if \(u\) has type \(\rch{\reg_1}{x}{\TRUE}{\rchepsilon{\reg_2}{z}{z<x}}\),
then the process
\(\inexp{u}{n}{r}{\soutatom{r}{n-1}}\) is allowed
but \(\inexp{u}{n}{r}{\soutatom{r}{n}}\) is not.

Type judgments for processes and expressions are
now of the form \( \env;\predenv;\chenv \p P \) and \( \env;\predenv;\chenv \p v\COL\mty \), where \( \predenv \) is a sequence of formulas.
Intuitively, \( \env;\predenv;\chenv \p P \) means that \( P \) is well-typed under the environments \( \env \) and \( \chenv \) assuming that all the formulas in \( \predenv \) holds.

The selected typing rules are shown in Figure~\ref{fig:refinement_type_system}.
The rules for the other constructs are identical to that of the simple type system; 
the complete list of typing rules appears in \ifaplas{the extended version~\cite{fullversion}}{Appendix~\ref{sec:refinement-apx}}. %
The rules shown in Figure~\ref{fig:refinement_type_system} are fairly standard rules for refinement type systems.
In \rn{RT-Out}, the notation
\(\predenv \vDash \pred\) means that \(\pred\) is a logical consequence of
\(\predenv\); for example, \(x<y, y<z \vDash x<z\) holds.
In the typing rules, we implicitly require that
all the type judgments are well-formed, in the sense that
all the integer variables occurring in a formula is
properly declared in \(\env\) or bound by a channel type constructor;
see \ifaplas{the extended version~\cite{fullversion}}{Appendix~\ref{sec:refinement-apx}}
for the well-formedness condition.
\begin{figure}[tb]
    \centering
    \small
    \begin{minipage}{\linewidth}
        \centering
        \begin{prooftree}
            \AxiomC{$\env;\predenv;\chenv\p x\COL\rch{\reg}{\seq{y}}{\pred}{\seq{\chty}}$}
            \AxiomC{$\env,\seq{y}\COL\seq{\ty}; \predenv,\pred; \chenv,\seq{z}\COL\seq{\chty} \p P$}
            \RightLabel{\textsc{(RT-In)}}
            \BinaryInfC{$\env;\predenv;\chenv\p \inexp{x}{\seq{y}}{\seq{z}}P$}
        \end{prooftree}
    \end{minipage}
    \\
    \begin{minipage}{0.93\linewidth}
        \infrule[RT-Out]
        {\env;\predenv;\chenv\p x\COL\rch{\reg}{\seq{y}}{\pred}{\seq{\chty}}\andalso
        \env;\predenv;\chenv\p \seq{v}\COL\seq{\ty}\andalso
        \predenv \vDash [\seq{v}/\seq{y}]\pred\\
        \env;\predenv;\chenv\p \seq{w}\COL[\seq{v}/\seq{y}]\seq{\chty}\andalso
        \env;\predenv;\chenv\p P}
        {\env;\predenv;\chenv \vdash \outexp{x}{\seq{v}}{\seq{w}}P}
    \end{minipage}
    \\
    \begin{minipage}{\linewidth}
        \centering
        \begin{prooftree}
            \AxiomC{$\env;\predenv;\chenv\p x\COL\rch{\reg}{\seq{y}}{\pred}{\seq{\chty}}$}
            \AxiomC{$\env,\seq{y}\COL\seq{\ty}; \predenv,\pred; \chenv,\seq{z}\COL\seq{\chty} \p P$}
            \RightLabel{\textsc{(RT-RIn)}}
            \BinaryInfC{$\env;\predenv;\chenv\p \rinexp{x}{\seq{y}}{\seq{z}}P$}
        \end{prooftree}
    \end{minipage}
    \\
    \begin{minipage}{\linewidth}
        \centering
        \begin{prooftree}
            \AxiomC{$\env; \predenv; \chenv \p v\COL\ty$}
            \AxiomC{$\env; \predenv, v \neq 0; \chenv \p P_1$}
            \AxiomC{$\env; \predenv, v =    0; \chenv \p P_2$}
            \RightLabel{\textsc{(RT-If)}}
            \TrinaryInfC{$\env;\predenv;\chenv\p \ifexp{v}{P_1}{P_2}$}
        \end{prooftree}
    \end{minipage}
    \\
    \begin{minipage}{0.45\linewidth}
        \centering
        \begin{prooftree}
            \AxiomC{$x \COL \chty \in \chenv$}
            \RightLabel{\textsc{(RT-Var-Ch)}}
            \UnaryInfC{$\env;\predenv;\chenv \p x\COL\chty$}
        \end{prooftree}
    \end{minipage}
    \normalsize
    \caption{\skchanged{Selected} typing rules of the refinement type system for the $\pi$-calculus}
    \label{fig:refinement_type_system}
\end{figure}

\subsection{Program Transformation}
\label{sec:rx}

Based on the refinement type system above,
we refine the transformation relation to
$\env; \predenv; \chenv \vdash P \Rightarrow \prog{\Def}{\Exp}$.
The only change is in the following rule for non-replicated inputs.\footnote{The rule for replicated inputs is also modified in a similar manner.}
\infrule[RX-In]
        {\env;\predenv;\chenv\p x\COL\rch{\reg}{\tilde{y}}{\pred}{\tilde{\chty}}
          \andalso
          \env,\seq{y}\COL\seq{\ty}; \predenv,\pred; \chenv,\seq{z}\COL\seq{\chty} \vdash P \Rightarrow \prog{\Def}{\Exp}}
        {\env;\predenv;\chenv \vdash \inexp{x}{\seq{y}}{\seq{z}}P
          \qquad\qquad\qquad\qquad\qquad\qquad\qquad\qquad         \\
          \Rightarrow
          \prog{\ndlet{y}{\textbf{Assume}(\pred);\Def}}{\ndlet{y}{\assexp{\pred}\Exp}}}
        Here, we insert \(\assatom{\pred}\), based on the refinement type
        of \(x\).
        The expression \(\ndlet{y}{\textbf{Assume}(\pred);\Exp}\) first
        instantiates \(\seq{y}\) to some integers in a non-deterministic manner,
        but proceeds to evaluate \(E\) only if the values of \(\seq{y}\) satisfy
        \(\pred\). Thus, the termination analysis for the target sequential program may assume that
       \(\seq{y}\) satisfies \(\pred\) in \(\Exp\).

\begin{example}  \label{ex:refined-transformation-for-f}
Let us explain how the process \( \DEC \) introduced in Example~\ref{ex:weakeness-of-basic-transformation} is translated by the refined translation.
Recall that the following simple types were assigned to the channels:
\begin{align*}
    \pre\COL \Chty{\reg_1}{\ty}{\sChty{\reg_2}{\ty}}, \quad
    \cname{f} \COL \Chty{\reg_3}{\ty}{\sChty{\reg_4}{\ty}}, \quad
    r \COL \sChty{\reg_4}{\ty}, \quad
    s \COL \sChty{\reg_2}{\ty}.
\end{align*}
By the refinement type system, the above types can be refined as:
\begin{align*}
    &\pre\COL \rch{\reg_1}{n}{\TRUE}{\sChty{\reg_2}{x; x < n}}, \quad
    \cname{f} \COL \rch{\reg_3}{n}{\TRUE}{\sChty{\reg_4}{x; \TRUE}}, \\
    &r \COL \sChty{\reg_4}{x; \TRUE}, \quad
    s \COL \sChty{\reg_2}{x; x < n}.
\end{align*}
For example, one can check that the output \( \soutatom{r}{n - 1} \) on the first line of \( \DEC \) is well-typed because \( \models [n - 1/x]x < n \) holds.
Note that this \( r \) is the variable bound by \( \inatom{\pre}{n}{r} \) and thus has the type \( \sChty{\reg_2}{x; x < n} \).

Therefore, by the rule \rn{RX-In}, the input \( \sinexp{s}{x}\outatom{f}{x}{r} \) is now translated as follows:
\infrule{\env;\predenv;\chenv\p s\COL\sChty{\reg_2}{x; x < n}
          \quad
          \env, x \COL \ty; \predenv,x < n; \chenv \vdash \outatom{f}{x}{r} \Rightarrow \prog{\Def}{\Fname{\reg_3}(x)}}
{\env;\predenv;\chenv \vdash \sinexp{s}{x}\outatom{f}{x}{r} \hspace{7cm}
          \qquad \\
          \Rightarrow
          \prog{(\ndlet*{x}{\textbf{Assume}(x < n);\Def})}{(\ndlet*{x}{\textbf{Assume}(x < n);\Fname{\reg_3}(x)})}}
with suitable \( \env \), \( \predenv \) and \( \chenv \).
By translating the whole process, we obtain
\begin{align*}
\Fname{\reg_3}(n) = \ &\textbf{if }{n<0} \textbf{ then } \Fname{\reg_4}(1) \\
    &\textbf{else }(\Fname{\reg_1}(n) \nondet \ndlet*{x}{\assexp{x < n}{\Fname{\reg_3}(x)}})
\end{align*}
as desired.
The other function definitions are given as in the case of Example~\ref{ex:weakeness-of-basic-transformation} (except for the fact that some redundant assertions \( \ndlet*{x}{\assatom{x < n}} \) are added).
\end{example}

\shchanged{
The soundness of the refined translation is obtained from the following argument.
We first extend the \(\pi\)-calculus with the $\ASSUME$ statement.
Then the refined translation can be decomposed into the following two steps:
(a) given a \(\pi\)-calculus process $P$,
insert $\ASSUME$ statements based on refinement types 
and obtain a process $P'$; and
(b) apply the translation of Section~\ref{sec:approach} to $P'$ (where $\ASSUME$ is just mapped to itself)
and obtain a sequential program $S$.
The soundness of step (b) follows by an easy modification of the proof\ifaplas{}{ in Appendix~\ref{sec:soundness}} 
for the basic transformation (just add the case for $\ASSUME$).
So, the termination of \(S\) would imply that of \(P'\).
Now, from the soundness of the refinement type system (which follows from a standard
argument on type preservation and progress), it follows that the $\ASSUME$ statements
inserted in step (a) always succeed. Thus,
the termination of \(P'\) would imply that of \(P\).
We can, therefore, conclude that if \(S\) is terminating, so is \(P\).
}
\newcommand{\constr}{\mathbb{C}}
\newcommand{\Inf}[4]{\fname{Inf}(#1; #2; #3; #4)}
\newcommand{\newty}[1]{\fname{NewTy}(#1)}
\newcommand{\newtys}[1]{\fname{NewTys}(#1)}
\newcommand{\subty}[3]{\fname{SubTy}(#1; #2; #3)}
\newcommand{\subtys}[3]{\fname{SubTys}(#1; #2; #3)}

\subsection{Type Inference}  \label{sec:inference}

This section discusses how to infer refinement types automatically
to automatically achieve the transformation.
As in refinement type inference for functional programs~\cite{Jhala08,Unno09PPDP,DBLP:journals/jar/ChampionCKS20},
we can reduce refinement type inference for the \(\pi\)-calculus to
the problem of CHC (Constrained Horn Clauses) solving~\cite{Bjorner15}.

We explain the procedure through an example.
Once again, we use the process \( \DEC \) introduced in Example~\ref{ex:weakeness-of-basic-transformation}.
We first perform type inference for the simple type system in Section~\ref{sec:targetlanguage}, and (as we have seen) obtain the following simple types for \( \pre \) and \( f \):
\begin{align*}
    \pre\COL \Chty{\reg_1}{\ty}{\sChty{\reg_2}{\ty}}, \quad
    \cname{f} \COL \Chty{\reg_3}{\ty}{\sChty{\reg_4}{\ty}}
\end{align*}
Here, we have omitted the types for other (bound) channels \(r,s,y\),
as they can be determined based on those of \( \pre \) and \( f \).
Based on the simple types, we prepare the following templates for refinement types.
\begin{align*}
  \pre\COL \rch{\reg_1}{n}{P_1(n)}{\sChty{\reg_2}{x; P_2(n,x)}}, \quad
  f\COL \rch{\reg_3}{n}{P_3(n)}{\sChty{\reg_4}{x; P_4(n,x)}}.
\end{align*}
Here, \(P_i\) (\(i\in\set{1,\ldots,4}\)) is a predicate variable that represents
unknown conditions.

Based on the refinement type system, we can generate the following constraints on
the predicate variables.
\[
\begin{array}{l}
  \forall n.(P_1(n) \imp P_2(n, n-1))\qquad
  \forall n.(P_3(n)\land n<0 \imp P_4(n, 1)) \\
  \forall n.(P_3(n)\land n\ge 0 \imp P_1(n-1))\\
  \forall n,x.(P_3(n)\land n\ge 0\land P_2(n-1,x) \imp P_3(x))\\
  \forall m.(\TRUE \imp P_3(m))
\end{array}
\]
Here, the first constraint comes from the first line of the process,
and the second constraint (the third and fourth constraints, resp.)
comes from the then-part (the else-part, resp.)
of the second line of the process. The last constraint
comes from \(\outatom{f}{m}{r}\).

The generated constraints are in general a set of \emph{Constrained Horn Clauses}
(CHCs)~\cite{Bjorner15} of the form
\(\forall \seq{x}.( P_1(\seq{v}_1)\land \cdots \land P_k(\seq{v}_k)\land \pred \imp H)\),
where \(P_1,\ldots,P_k\) are predicate variables, \(\pred\) is a formula
of integer arithmetic (without predicate variables),
and \(H\) is either of the form \(P(\seq{v})\) or \(\pred'\).
The problem of finding a solution (i.e. an assignment of predicates to
predicate variables) of a set of CHCs is undecidable in general,
but there are various automated tools (called CHC solvers)
for solving the problem~\cite{DBLP:journals/fmsd/KomuravelliGC16,DBLP:journals/jar/ChampionCKS20}.
Thus, by using such a CHC solver, we can solve the constraints on predicate variables,
and obtain refinement types by substituting the solution for the templates of
refinement types.

For the example above, the following is a solution.
\[
\begin{array}{l}
  P_1(n)\equiv \TRUE\qquad P_2(n,x)\equiv x < n\qquad 
P_3(x) \equiv \TRUE \qquad P_4(n,x)\equiv \TRUE.
\end{array}
\]
This is exactly the predicates we used in Example~\ref{ex:refined-transformation-for-f} to translate \( \DEC \) using the refined approach.

\subsubsection*{Adding extra CHCs.}
Actually, a further twist is necessary in the step of CHC solving.
As in the example above, all the CHCs generated based on the refinement typing rules
are of the form \(\cdots \imp P_i(\seq{v})\) (i.e., the head of every CHC is
an atomic formula on a predicate variable).
Thus, there always exists a trivial solution for the CHCs, which instantiates
all the predicate variables to \(\TRUE\).
For the example above,
\[
\begin{array}{l}
  P_1(n)\equiv \TRUE\qquad P_2(n,x)\equiv \TRUE\qquad 
P_3(n) \equiv \TRUE \qquad P_4(n,x)\equiv \TRUE
\end{array}
\]
is also a solution,
but using the trivial solution,
our transformation yields the non-terminating program.
This program is essentially the same as the one in Example~\ref{ex:weakeness-of-basic-transformation} since \( \ndlet*{x}{\assexp{\TRUE}}\Exp \) is equivalent to \( \ndlet*{x}\Exp \).
Typical CHC solvers indeed tend to find the trivial solution.

To remedy the problem above, in addition to the CHCs generated from the typing rules,
we add extra constraints that prevent infinite loops.
For the example above, the definition of \(\Fname{\reg_3} \) (which corresponds to the channel \( f \)) in the translated program is of the form
\[  \Fname{\reg_3}(n) = \ifexp{\!n<0\!}{\!\skipexp\!}{\Fname{\reg_1}(n) \nondet  (\letexp{x}{\ndint}\assexp{P_2(n,x)}\Fname{\reg_3}(x))}.
\]
Thus we add the clause:
\[ P_2(n,x) \imp n\ne x\]
to prevent an infinite loop \(\Fname{\reg_3}(m)\red \Fname{\reg_3}(m) \red \cdots\).
With the added clause, a CHC solver \hoice{}~\cite{DBLP:journals/jar/ChampionCKS20}
indeed returns \(n<x\) as the solution for \(P_2(n,x)\).

In general, we can add the extra CHCs in the following, counter-example-guided manner.
\begin{enumerate}
\item \(\mathcal{C} := \) the CHCs generated from the typing rules
\item \(\theta := \mathit{callCHCsolver}(\mathcal{C})\)
\item \(S := \) the sequential program generated based on the solution \(\theta\)
\item if \(S\) is terminating then return OK; otherwise, 
 analyze \( S \) to find an infinite reduction sequence, add an extra clause to \(\mathcal{C}\) to disable the infinite sequence, and go back to 2.
\end{enumerate}
\changed{
More precisely, in the last step, the backend termination analysis tool generates a lasso
as a certificate of non-termination. We extract 
a chain $f(\seq{x}) \to \dots \to f(\seq{x}')$ of recursive calls
from the lasso, and
add an extra clause requiring $\seq{x} \neq \seq{x}'$
to \(\mathcal{C}\). This is naive and insufficient
for excluding out an infinite sequence like \(f(1)\to f(2) \to f(3) \to \cdots\).
We plan to refine the method by incorporating more sophisticated techniques
developed for sequential programs~\cite{hashimoto2015refinement}.
}
\section{Implementation and Preliminary Experiments}  \label{sec:implementation}

\subsection{Implementation}
We have implemented a termination analysis tool for the $\pi$-calculus based on the method
described in Sections~\ref{sec:approach} and \ref{sec:refinement}.
This tool was written in OCaml. We chose C language as the actual target of our translation,
and used \ult{}~\cite{heizmann2016ultimate} (version 0.2.1) as a termination analysis tool for C.

For the refinement type inference described in Section~\ref{sec:inference},
we have used
\hoice{}~\cite{DBLP:journals/jar/ChampionCKS20} (version 1.8.3)
and \zthree{}~\cite{de2008z3} (version 4.8.10) as backend CHC solvers.
Since a stronger solution for CHCs is preferable as discussed at the end of
Section~\ref{sec:inference}, if \hoice{} and \zthree{} return
different solutions
\(\set{\predvar_1\mapsto \pred_1,\ldots,\predvar_n\mapsto \pred_n}\)
and 
\(\set{\predvar_1\mapsto \pred'_1,\ldots,\predvar_n\mapsto \pred'_n}\),
then we used the solution
\(\set{\predvar_1\mapsto \pred_1\land \pred'_1,\ldots,\predvar_n\mapsto \pred_n\land \pred'_n}\)
for inserting \textbf{Assume} commands.

To make the analysis precise, the implementation is actually based on an extension of
the refinement type system in Section~\ref{sec:rtype} with subtyping; 
see \ifaplas{the extended version~\cite{fullversion}}{Appendix~\ref{sec:subtyping}}.

\subsection{Preliminary Experiments}
\label{sec:experiments}

We prepared a collection of $\pi$-calculus processes,
and ran our tool on them.
Our experiment was conducted on Intel Core i7-10850H CPU with 32GB memory.
For comparison, we have also run the termination analysis mode of
\typical{}~\cite{TyPiCal,KobayashiHybrid} on the same instances.

\newcolumntype{Y}{>{\centering\arraybackslash}p{5em}}
\newcolumntype{Z}{>{\centering\arraybackslash}p{12em}}
\begin{table}[tbp]
    \centering
    \caption{Results of the experiments}
    \begin{tabular}{ZYYY}
        \hline
        Test case & Basic & Refined & \typical{} \\ \hline \hline
        client-server   &  2.5 &  2.7 & 0.006 \\
        stateful-server-client & FAIL & FAIL & 0.006 \\
        parallel-or     &  2.4 &  2.9 & 0.006 \\
        broadcast       &  3.6 &  3.3 & 0.004 \\
        btree           & FAIL & FAIL & 0.011 \\
        stable          & FAIL & FAIL & 0.003 \\
        ds-ex5-1        & FAIL & FAIL & 0.002 \\
        factorial       &  3.9 &  4.4 & 0.002 \\
        ackermann       & 22.4 & 26.0 & 0.003 \\
        \hline
        fibonacci       &  4.8 &  4.4 & 0.003 \\
        even/odd        &  7.0 &  7.6 & 0.002 \\
        factorial-pred  & FAIL & 28.2 & FAIL  \\
        fibonacci-pred  & FAIL & 28.2 & FAIL  \\ 
        even/odd-pred   & FAIL & 10.1 & FAIL  \\
        sum-neg         &  7.6 & 13.1 & FAIL  \\
        upperbound      &  3.8 &  3.9 & FAIL  \\
        nested-replicated-input1    &  2.3 &  2.4 & FAIL  \\
        nested-replicated-input2    & FAIL & FAIL & FAIL  \\
        nested-replicated-input3    &  3.7 &  4.0 & 0.010 \\
        deadlock        & FAIL &  2.9 & FAIL  \\
        \hline
    \end{tabular}
    \label{tab:result}
\end{table}

The experimental results are summarized in Table~\ref{tab:result}.
The columns ``Basic'' and ``Refined'' show the results for
the basic method in Section~\ref{sec:approach}
and the refined method in Section~\ref{sec:refinement} respectively.
The numbers show the running times measured in seconds, and
``FAIL'' means that the verification failed due to the incompleteness of
the reduction; non-terminating sequential programs were obtained in those cases.
The column ``\typical{}'' shows the analogous result for \typical{}.
The termination analysis of \typical{} roughly depends on Deng and Sangiorgi's
method~\cite{Deng06IC}. ``FAIL'' in the column means that the process does not
satisfy the (sufficient) conditions for termination~\cite{Deng06IC}.
The termination analysis of \typical{} treats numbers as natural numbers,
and is actually unsound in the presence of arbitrary integers
(for example, \(\outatom{f}{m}{r}\PAR
\rinexp{f}{x}{r}{\ifexp{x=0}{\soutatom{r}{1}}{\outatom{f}{x-1}{r}}}\) is
judged to be terminating for any \(m\)).

The test cases consist of two categories. The first one, shown above
the horizontal line, has been taken from the sample programs of \typical{}.
Among them, we have excluded out those that are not
related to termination analysis (note that \typical{} can perform deadlock/lock-freedom
analysis and information flow analysis besides termination analysis).
The second category, shown below the horizontal line, consists of those prepared by
us,\footnote{Unfortunately, there are no
  standard benchmark set for the termination analysis for the \(\pi\)-calculus.}
including the samples discussed in the paper.
All the processes in the test cases are terminating.

For \testcase{stateful-server-client}, \testcase{btree}, \testcase{stable},
and \testcase{ds-ex5-1} in the first category, and \testcase{nested-replicated-input2}
in the second category,
our analysis fails for essentially the same reason.
The following is a simplified version of \testcase{ds-ex5-1}:
\[
\soutatom{a}{} \PAR \soutatom{b}{} \PAR \srinexp{a}{}\sinexp{b}{}\soutatom{a}{}.
\]
The process above is terminating because each run of the third process consumes
a message on \(b\). Our reduction however ignores communications on \(b\) and produces
the following non-terminating program:
\[
(\set{\Fname{\reg_a}()=\Fname{\reg_a}(), \Fname{\reg_b}()=\skipexp}, \Fname{\reg_a}()\nondet
\Fname{\reg_b}()).
\]

For the second category, our refined method clearly outperforms the basic method
and \typical{}.
We explain some of the test cases in the second category.
The test cases \testcase{fibonacci} and \testcase{nested-replicated-input3} are from Example~\ref{ex:fib}
and \ref{ex:nested_rep} respectively, and \testcase{even/odd} is a mutually recursive process
that judges whether a given number is even or odd.
The process \testcase{deadlock} is the following one:
\[    \srinexp{\cname{loop}}{} \soutatom{\cname{loop}}{} 
    \PAR \sinexp{r}{} \soutatom{\cname{loop}}{}. \]
    This process is terminating, because the subprocess
    \(\sinexp{r}{} \soutatom{\cname{loop}}{}\) is blocked forever, without ever sending
    a message to \(\cname{loop}\).
    With the refinement type system, the channel \(r\) is given type:
    \(\sChty{\reg}{\epsilon;\FALSE}\), and
    \(\sinexp{r}{} \soutatom{\cname{loop}}{}\) is translated to:
    \[
\letexp{\epsilon}{\epsilon}{\assexp{\FALSE}{\Fname{\reg_{\cname{loop}}}()}},
    \]
    which is terminating by $\ASSUME(\FALSE)$.
    The process \testcase{upperbound} is the following process:
    \begin{align*}
    \soutatom{f}{0} \PAR \srinexp{f}{x} \ifexp{x>10}{\zeroexp}{\soutatom{f}{x+1}}.
\end{align*}
    It is terminating because the argument of \(f\) monotonically increases, and is bounded
    above by \(10\). \typical{} cannot make such reasoning.

\section{Related Work}  \label{sec:relatedwork}

As mentioned in Section~\ref{sec:introduction}, there have been a number of
studies on termination of the \(\pi\)-calculus~\cite{Deng06IC,Demangeon07,SangiorgiTermination,KobayashiHybrid,Yoshida04IC,DBLP:journals/jlp/DemangeonHS10,Venet98SAS},
but most of them have been rather theoretical, and few tools have been developed.
Our technique has been partially inspired by Deng and Sangiorgi's work~\cite{Deng06IC}, especially by their observation that a process is terminating just if
there is no infinite chain of communications on replicated input processes.
Deng and Sangiorgi ensured the lack of infinite chains by using a type system.
They actually proposed four system, a core system and three kinds of extensions.
Our approach roughly corresponds to the first extension of their system (\cite{Deng06IC}, Section 4), which requires that, in every
chain of communications, the values of messages
monotonically decrease. An advantage of our approach is that we can use mature tools
for sequential programs to reason about how the values of messages change.
Our approach does not subsume the second and third extensions of Deng and Sangiorgi's
system, which take into account synchronizations over multiple channels; it is
left for future work to study whether such extensions can be incorporated in
our approach.

To our knowledge,
\typical{}~\cite{TyPiCal,KobayashiHybrid}
is the only automated termination analysis tool.
\typical{}'s termination analysis is based on Deng and Sangiorgi's
method~\cite{Deng06IC}, but is quite limited in reasoning about
the values sent along channels; it only considers natural numbers, and
the ordering on them is limited to the standard order on natural numbers.
Thus, for example, \typical{} cannot prove the termination of the process \testcase{upperbound}
as described in Section~\ref{sec:implementation}.

Recently, there have been studies on type systems for estimating
the (time) complexity of processes for the
\(\pi\)-calculus~\cite{DBLP:conf/esop/BaillotG21,DBLP:conf/concur/BaillotG021} and related session
calculi~\cite{DBLP:conf/lics/Das0P18,DBLP:journals/pacmpl/Das0P18}.
Since the existence of a finite upper-bound implies termination, those analyses
can, in principle, be used also for reasoning about termination, but the resulting
termination analysis would be too conservative.
It would be interesting to investigate whether our approach of reduction to 
sequential programs can be extended to achieve complexity analysis for the
\(\pi\)-calculus.
\changed{Refinement types for variants of the \(\pi\)-calculus have been
studied before~\cite{griffith2013liquidpi,baltazar2012linearly}.
Our contribution in this regard is the application to
termination analysis.}

Cook et al.~\cite{Cook07PLDI} proposed a method for proving
termination of multi-threaded programs. Their technique also makes use
of a termination tool for sequential programs.  As their language
model is quite different from ours (they deal with imperative programs
with shared memory and locks, rather than message-passing programs),
however, their method is quite different from ours.

\section{Conclusion}  \label{sec:conclusion}

We have proposed a method for reducing termination verification for
the \(\pi\)-calculus to that for sequential programs and implemented
an automated termination analysis tool based on the method.
Our approach allows us to reuse powerful termination analysis tools
developed for sequential programs.

Future work includes (i) a further refinement of our reduction
and (ii) applications of our method to other message-passing-style
concurrent programming languages.
As for the first point, there are a few known limitations in the current reduction.
Besides the issues mentioned at the end of Example~\ref{ex:nested_rep}
and Section~\ref{sec:implementation},
there is a limitation that channels of the same region are merged to the same function,
which leads to the loss of precision.
For example, consider:
\begin{align*}
  &\srinexp{c}{x} \ifexp{x < 0}{\zeroexp }{\soutatom{c}{x - 1}} \\
&  \PAR \srinexp{d}{x} \ifexp{x > 0}{\zeroexp }{\soutatom{d}{x + 1}} \\
&\PAR \soutatom{e}{c} \PAR \soutatom{e}{d} \PAR \soutatom{c}{0}
\end{align*}
The process is terminating, but our approach fails to prove it. Since the same region
is assigned to \(c\) and \(d\) (because both are sent along \(e\)), the replicated input processes
are translated to non-deterministic function definitions:
\begin{align*}
    \Fname{\reg}(x) &= \ifexp{x < 0}{\skipexp}{\Fname{\reg}(x - 1)} \\
    \Fname{\reg}(x) &= \ifexp{x > 0}{\skipexp}{\Fname{\reg}(x + 1)}, 
\end{align*}
which cause an infinite reduction \(\Fname{\reg}(0) \red \Fname{\reg}(-1)\red \Fname{\reg}(0)\red \cdots\).
One remedy to this problem would be to introduce region polymorphism and translate
processes to higher-order functional programs.

\subsubsection*{Acknowledgments}
We would like to thank anonymous referees for useful comments.
This work was supported by JSPS KAKENHI Grant Number JP20H05703.

\bibliographystyle{splncs04}
\bibliography{ref,abbrv,koba}

\begin{thebibliography}{10}
\providecommand{\url}[1]{\texttt{#1}}
\providecommand{\urlprefix}{URL }
\providecommand{\doi}[1]{https://doi.org/#1}

\bibitem{DBLP:conf/esop/BaillotG21}
Baillot, P., Ghyselen, A.: Types for complexity of parallel computation in
  pi-calculus. In: Proceedings of {ESOP} 2021. LNCS, vol. 12648, pp. 59--86.
  Springer (2021). \doi{10.1007/978-3-030-72019-3\_3}

\bibitem{DBLP:conf/concur/BaillotG021}
Baillot, P., Ghyselen, A., Kobayashi, N.: Sized types with usages for parallel
  complexity of pi-calculus processes. In: Proceedings of {CONCUR} 2021.
  LIPIcs, vol.~203, pp. 34:1--34:22. Schloss Dagstuhl - Leibniz-Zentrum
  f{\"{u}}r Informatik (2021). \doi{10.4230/LIPIcs.CONCUR.2021.34}

\bibitem{baltazar2012linearly}
Baltazar, P., Mostrous, D., Vasconcelos, V.T.: Linearly refined session types.
  Electronic Proceedings in Theoretical Computer Science  \textbf{101},
  38–49 (2012). \doi{10.4204/eptcs.101.4}

\bibitem{Bjorner15}
Bj{\o}rner, N., Gurfinkel, A., McMillan, K.L., Rybalchenko, A.: Horn clause
  solvers for program verification. In: Fields of Logic and Computation {II} -
  Essays Dedicated to Yuri Gurevich on the Occasion of His 75th Birthday. LNCS,
  vol.~9300, pp. 24--51. Springer (2015). \doi{10.1007/978-3-319-23534-9\_2}

\bibitem{DBLP:journals/jar/ChampionCKS20}
Champion, A., Chiba, T., Kobayashi, N., Sato, R.: Ice-based refinement type
  discovery for higher-order functional programs. J. Autom. Reason.
  \textbf{64}(7),  1393--1418 (2020). \doi{10.1007/s10817-020-09571-y}

\bibitem{Cook07PLDI}
Cook, B., Podelski, A., Rybalchenko, A.: Proving thread termination. In:
  Proceedings of {PLDI} 2007. pp. 320--330. ACM Press (2007).
  \doi{10.1145/1250734.1250771}

\bibitem{DBLP:journals/cacm/CookPR11}
Cook, B., Podelski, A., Rybalchenko, A.: Proving program termination. Commun.
  {ACM}  \textbf{54}(5),  88--98 (2011). \doi{10.1145/1941487.1941509}

\bibitem{DBLP:journals/pacmpl/Das0P18}
Das, A., Hoffmann, J., Pfenning, F.: Parallel complexity analysis with temporal
  session types. Proc. {ACM} Program. Lang.  \textbf{2}({ICFP}),  91:1--91:30
  (2018). \doi{10.1145/3236786}

\bibitem{DBLP:conf/lics/Das0P18}
Das, A., Hoffmann, J., Pfenning, F.: Work analysis with resource-aware session
  types. In: Proceedings of {LICS} 2018. pp. 305--314. {ACM} (2018).
  \doi{10.1145/3209108.3209146}

\bibitem{de2008z3}
De~Moura, L., Bj{\o}rner, N.: Z3: An efficient {SMT} solver. In: Proceedings of
  {TACAS} 2008. pp. 337--340. Springer (2008).
  \doi{10.1007/978-3-540-78800-3\_24}

\bibitem{Demangeon07}
Demangeon, R., Hirschkoff, D., Kobayashi, N., Sangiorgi, D.: On the complexity
  of termination inference for processes. In: Proceedings of TGC 2007. LNCS,
  vol.~4912, pp. 140--155. Springer (2008). \doi{10.1007/978-3-540-78663-4\_11}

\bibitem{DBLP:journals/jlp/DemangeonHS10}
Demangeon, R., Hirschkoff, D., Sangiorgi, D.: Termination in higher-order
  concurrent calculi. J. Log. Algebraic Methods Program.  \textbf{79}(7),
  550--577 (2010). \doi{10.1016/j.jlap.2010.07.007}

\bibitem{Deng06IC}
Deng, Y., Sangiorgi, D.: Ensuring termination by typability. Info. Comput.
  \textbf{204}(7),  1045--1082 (2006). \doi{10.1016/j.ic.2006.03.002}

\bibitem{freqterm}
Fedyukovich, G., Zhang, Y., Gupta, A.: Syntax-guided termination analysis. In:
  Proceedings of CAV 2018. LNCS, vol. 10981, pp. 124--143. Springer (2018).
  \doi{10.1007/978-3-319-96145-3\_7}

\bibitem{griffith2013liquidpi}
Griffith, D., Gunter, E.L.: {LiquidPi}: Inferrable dependent session types. In:
  NASA Formal Methods Symposium. pp. 185--197. Springer (2013).
  \doi{10.1007/978-3-642-38088-4\_13}

\bibitem{hashimoto2015refinement}
Hashimoto, K., Unno, H.: Refinement type inference via horn constraint
  optimization. In: International Static Analysis Symposium. pp. 199--216.
  Springer (2015). \doi{10.1007/978-3-662-48288-9\_12}

\bibitem{heizmann2016ultimate}
Heizmann, M., Dietsch, D., Greitschus, M., Leike, J., Musa, B., Sch{\"a}tzle,
  C., Podelski, A.: Ultimate automizer with two-track proofs. In: Proceedings
  of TACAS 2016. pp. 950--953. Springer (2016).
  \doi{10.1007/978-3-662-49674-9\_68}

\bibitem{TyPiCal}
Kobayashi, N.: {\sc TyPiCal}: A type-based static analyzer for the pi-calculus.
  Tool available at \url{https://www-kb.is.s.u-tokyo.ac.jp/~koba/typical/}
  (2005)

\bibitem{KobayashiHybrid}
Kobayashi, N., Sangiorgi, D.: A hybrid type system for lock-freedom of mobile
  processes. ACM Trans. Prog. Lang. Syst.  \textbf{32}(5),  1--49 (2010).
  \doi{10.1145/1745312.1745313}

\bibitem{DBLP:journals/fmsd/KomuravelliGC16}
Komuravelli, A., Gurfinkel, A., Chaki, S.: {SMT}-based model checking for
  recursive programs. Formal Methods Syst. Des.  \textbf{48}(3),  175--205
  (2016). \doi{10.1007/s10703-016-0249-4}

\bibitem{Kuwahara2014Termination}
Kuwahara, T., Terauchi, T., Unno, H., Kobayashi, N.: Automatic termination
  verification for higher-order functional programs. In: Proceedings of {ESOP}
  2014. LNCS, vol.~8410, pp. 392--411. Springer (2014).
  \doi{10.1007/978-3-642-54833-8\_21}

\bibitem{milner1993polyadic}
Milner, R.: The polyadic $\pi$-calculus: a tutorial. Logic and algebra of
  specification pp. 203--246 (1993). \doi{10.1007/978-3-642-58041-3\_6}

\bibitem{Pierce96MSCS}
Pierce, B., Sangiorgi, D.: Typing and subtyping for mobile processes.
  Mathematical Structures in Computer Science  \textbf{6}(5),  409--454 (1996).
  \doi{10.1017/S096012950007002X}

\bibitem{DBLP:conf/lics/PodelskiR04}
Podelski, A., Rybalchenko, A.: Transition invariants. In: Proceedings of {LICS}
  2004. pp. 32--41 (2004). \doi{10.1109/LICS.2004.1319598}

\bibitem{Jhala08}
Rondon, P.M., Kawaguchi, M., Jhala, R.: Liquid types. In: PLDI 2008. pp.
  159--169 (2008). \doi{10.1145/1375581.1375602}

\bibitem{SangiorgiTermination}
Sangiorgi, D.: Termination of processes. Mathematical Structures in Computer
  Science  \textbf{16}(1),  1--39 (2006). \doi{10.1017/S0960129505004810}

\bibitem{Unno09PPDP}
Unno, H., Kobayashi, N.: Dependent type inference with interpolants. In:
  Proceedings of {PPDP} 2009. pp. 277--288. {ACM} (2009).
  \doi{10.1145/1599410.1599445}

\bibitem{Venet98SAS}
Venet, A.: Automatic determination of communication topologies in mobile
  systems. In: Proceedings of {SAS}'98. LNCS, vol.~1503, pp. 152--167. Springer
  (1998). \doi{10.1007/3-540-49727-7\_9}

\bibitem{Yoshida04IC}
Yoshida, N., Berger, M., Honda, K.: Strong normalisation in the pi-calculus.
  Info. Comput.  \textbf{191}(2),  145--202 (2004).
  \doi{10.1016/j.ic.2003.08.004}

\bibitem{DBLP:conf/concur/YoshidaH99}
Yoshida, N., Hennessy, M.: Subtyping and locality in distributed higher order
  processes (extended abstract). In: Proceedings of {CONCUR} '99. LNCS,
  vol.~1664, pp. 557--572. Springer (1999). \doi{10.1007/3-540-48320-9\_38}

\end{thebibliography}
\ifaplas{}{
\clearpage
\appendix
\section{Operational Semantics} \label{sec:operational_semantics}
\subsection{Reduction Semantics of the \( \pi \)-Calculus}
We define a reduction relation~\cite{milner1993polyadic} as the operational semantics of the $\pi$-calculus.

As usual, we first define the structural congruence relation \( \piequiv \) on the set of processes.
\begin{definition}[structural congruence for processes]
  The \emph{structural congruence relation} $\piequiv$ on %
  $\pi$-calculus processes is defined as the least congruence relation that satisfies the following %
  rules.
    \begin{gather*}
        P_1 \mid P_2 \piequiv P_2 \mid P_1
        \qquad (P_1 \mid P_2) \mid P_3 \piequiv P_1 \mid (P_2 \mid P_3) 
        \\     P \mid \textbf{0} \piequiv P 
        \qquad (\nu x)\,\textbf{0} \piequiv \textbf{0}
        \qquad (\nu x)(\nu y)P \piequiv (\nu y)(\nu x)P
        \\     (\nu x)(P_1 \mid P_2) \piequiv P_1 \mid (\nu x)P_2 \quad \text{if $x$ does not freely occur in $P_1$}
    \end{gather*}
\end{definition}

Next, we define the reduction relation on processes. %
\begin{definition}
  The \emph{reduction relation $\to$ on %
    processes} is defined by the set of rules in Figure~\ref{fig:pi_reduction}.
    We write $\to^*$ and \( \to^+ \)for the reflexive transitive closure and the transitive closure of the reduction relation $\to$, respectively.
\end{definition}
\begin{figure}[tbp]
    \centering
    \small
    \begin{minipage}{\linewidth}
        \centering
        \begin{prooftree}
            \AxiomC{$\len{\seq{y}} = \len{\seq{v}}$}
            \AxiomC{$\len{\seq{z}} = \len{\seq{w}}$}
            \AxiomC{$\seq{v} \Downarrow \seq{i}$}
            \RightLabel{\textsc{(R-Comm)}}
            \TrinaryInfC{$\inexp{x}{\seq{y}}{\seq{z}}P_1 \PAR  \outexp{x}{\seq{v}}{\seq{w}}P_2 \red [\seq{i}/ \seq{y}, \seq{w} / \seq{z} ]P_1 \PAR P_2$}
        \end{prooftree}
    \end{minipage}
    \begin{minipage}{0.48\linewidth}
        \centering
        \begin{prooftree}
            \AxiomC{$P_1 \red P_1'$}
            \RightLabel{\textsc{(R-Par)}}
            \UnaryInfC{$P_1 \PAR P_2 \red P_1' \PAR P_2$}
        \end{prooftree}
    \end{minipage}
    \begin{minipage}{0.5\linewidth}
        \centering
        \begin{prooftree}
            \AxiomC{$P \red P'$}
            \RightLabel{\textsc{(R-Nu)}}
            \UnaryInfC{$\nuexp{x \COL \chty} P \red \nuexp{x \COL \chty}P' $}
        \end{prooftree}
    \end{minipage}
    \begin{minipage}{\linewidth}
        \centering
        \begin{prooftree}
            \AxiomC{$\len{\seq{y}} = \len{\seq{v}}$}
            \AxiomC{$\len{\seq{z}} = \len{\seq{w}}$}
            \AxiomC{$\seq{v} \Downarrow \seq{i}$}
            \RightLabel{\textsc{(R-RComm)}}
            \TrinaryInfC{$\rinexp{x}{\seq{y}}{\seq{z}}P_1 \PAR  \outexp{x}{\seq{v}}{\seq{w}}P_2 \red \rinexp{x}{\seq{y}}{\seq{z}}P_1 \PAR [\seq{i}/ \seq{y}, \seq{w} / \seq{z}]P_1 \PAR P_2$}
        \end{prooftree}
    \end{minipage}
    \begin{minipage}{\linewidth}
        \centering
        \begin{prooftree}
            \AxiomC{$v \Downarrow i \neq 0$}
            \RightLabel{\textsc{(R-If-T)}}
            \UnaryInfC{$\ifexp{v}{P_1}{P_2} \red P_1$}
        \end{prooftree}
    \end{minipage}
    \begin{minipage}{\linewidth}
        \centering
        \begin{prooftree}
            \AxiomC{$v \Downarrow 0$}
            \RightLabel{\textsc{(R-If-F)}}
            \UnaryInfC{$\ifexp{v}{P_1}{P_2} \red P_2$}
        \end{prooftree}
    \end{minipage}
    \begin{minipage}{\linewidth}
        \centering
        \begin{prooftree}
            \AxiomC{\( \len{\seq{x}} = \len{\seq{i}}\)}
            \RightLabel{\textsc{(R-LetND)}}
            \UnaryInfC{$\ndlet{x}{P} \red [\seq{i} / \seq{x}]P$}
        \end{prooftree}
    \end{minipage}
    \begin{minipage}{\linewidth}
        \centering
        \begin{prooftree}
            \AxiomC{$P \piequiv P_1 \red P_1' \piequiv P'$}
            \RightLabel{\textsc{(R-Cong)}}
            \UnaryInfC{$P \red P'$}
        \end{prooftree}
    \end{minipage}
    \begin{minipage}{0.3\linewidth}
        \centering
        \begin{prooftree}
            \AxiomC{}
            \RightLabel{\textsc{(R-Int)}}
            \UnaryInfC{$i \Downarrow i$}
        \end{prooftree}
    \end{minipage}
    \begin{minipage}{0.38\linewidth}
        \centering
        \begin{prooftree}
            \AxiomC{$\seq{v} \Downarrow \seq{i}$}
            \RightLabel{\textsc{(R-Op)}}
            \UnaryInfC{$\op(\seq{v}) \Downarrow \llbracket \op\rrbracket (\seq{i})$.}
        \end{prooftree}
    \end{minipage}
    \normalsize
    \caption{The reduction rules of the $\pi$-calculus. Here \( \llbracket \op \rrbracket \colon \mathbb{Z}^n \to \mathbb{Z} \) represents the interpretation of the operation \( \op \) whose arity is \( n \). }
    \label{fig:pi_reduction}
\end{figure}

\subsection{Reduction Semantics of the Sequential Language}
Here, we define the reduction semantics for the sequential language.
We actually define two kinds of semantics:
one is a standard reduction relation
\((\Def,\Exp)\sred (\Def',\Exp')\), which evaluates \( \Exp_1 \nondet \Exp_2\) to either \( \Exp_1 \) or \( \Exp_2\); the other is a non-standard reduction relation
\((\Def,\Exp)\nsred(\Def',\Exp')\), which
 does not discard branches of non-deterministic choices.

\begin{definition}
  The \emph{reduction relation $\seqto$ on %
    sequential programs} is defined by the set of rules in Figure~\ref{fig:seq_reduction}.
    In the rule \rn{SR-App} we are considering \(\Def\) as a map that maps \( f \) to \( \Def(f) =  \{ \lambda \seq{x}. \Exp \mid \fdef{f}{\seq{x}}{\Exp} \in \Def \}\).
\end{definition}
\begin{figure}[tb]
    \centering
    \small
    \begin{minipage}{\linewidth}
        \centering
        \begin{prooftree}
            \AxiomC{$\len{\seq{x}} = \len{\seq{i}}$}
            \RightLabel{\textsc{(SR-LetND)}}
            \UnaryInfC{$(\Def, \ndlet{x}{\Exp} )  \sred (\Def, [\seq{i}/\seq{x}]\Exp)$}
        \end{prooftree}
    \end{minipage}
    \begin{minipage}{\linewidth}
        \centering
        \begin{prooftree}
            \AxiomC{\( (\lambda \seq{y}.\Exp) \in  \Def(f)  \)}
            \AxiomC{$\len{\seq{y}} = \len{\seq{v}}$}
            \AxiomC{$\seq{v} \Downarrow \seq{i}$}
            \RightLabel{\textsc{(SR-App)}}
            \TrinaryInfC{$(\Def, f(\seq{v})) \sred (\Def, [\seq{i}/\tilde{y}] \Exp)$}
        \end{prooftree}
    \end{minipage}
    \begin{minipage}{\linewidth}
        \centering
        \begin{prooftree}
            \AxiomC{$v \Downarrow i$ \qquad \( i \neq 0\)}
            \RightLabel{\textsc{(SR-If-T)}}
            \UnaryInfC{$(\Def, \ifexp{v}{\Exp_1}{\Exp_2})  \sred (\Def, \Exp_1)$}
        \end{prooftree}
    \end{minipage}
    \begin{minipage}{\linewidth}
        \centering
        \begin{prooftree}
            \AxiomC{$v \Downarrow 0$}
            \RightLabel{\textsc{(SR-If-F)}}
            \UnaryInfC{$(\Def, \ifexp{v}{\Exp_1}{\Exp_2}) \sred (\Def, \Exp_2)$}
        \end{prooftree}
    \end{minipage}
    \\[.3cm]
    \begin{minipage}{\linewidth}
        \centering
        \begin{prooftree}
            \AxiomC{}
            \RightLabel{\textsc{(SR-Cho-L)}}
            \UnaryInfC{$(\Def, \Exp_1 \nondet \Exp_2) \sred (\Def, \Exp_1)$}
        \end{prooftree}
    \end{minipage}
    \\[.3cm]
    \begin{minipage}{\linewidth}
        \centering
        \begin{prooftree}
            \AxiomC{}
            \RightLabel{\textsc{(SR-Cho-R)}}
            \UnaryInfC{$(\Def, \Exp_1 \nondet \Exp_2) \sred (\Def, \Exp_2)$}
        \end{prooftree}
    \end{minipage}
    \begin{minipage}{\linewidth}
        \centering
        \begin{prooftree}
            \AxiomC{$v \Downarrow i$ \qquad \( i \neq 0\)}
            \RightLabel{\textsc{(SR-Ass-T)}}
            \UnaryInfC{$(\Def, \textbf{Assume}(v);E) \sred (\Def, \Exp)$}
        \end{prooftree}
    \end{minipage}
    \begin{minipage}{\linewidth}
        \centering
        \begin{prooftree}
            \AxiomC{$v \Downarrow 0$}
            \RightLabel{\textsc{(SR-Ass-F)}}
            \UnaryInfC{$(\Def, \textbf{Assume}(v);\Exp) \sred (\Def, \skipexp)$}
        \end{prooftree}
    \end{minipage}
    \normalsize
    \caption{Reduction rules of the sequential language}
    \label{fig:seq_reduction}
\end{figure}

We now define a non-standard reduction relation that keeps
all the non-deterministic branches \shchanged{during} the reduction.
This non-standard reduction relation has a better match with the reduction of processes.
Since processes have structural rules, we also introduce structural rules on expressions.

\begin{definition}[structural congruence for sequential expressions]
    The \emph{structural congruence relation for expressions}, written \( \Exp_1 \expequiv \Exp_2 \), is defined as the least congruence relation that satisfies the following rules.
    \begin{gather*}
        \Exp_1 \nondet \Exp_2 \expequiv \Exp_2 \nondet \Exp_1
        \qquad (\Exp_1 \nondet \Exp_2) \nondet \Exp_3 \expequiv \Exp_1 \nondet (\Exp_2 \nondet \Exp_3)
        \qquad \Exp \nondet \skipexp \expequiv \Exp
        \qquad
    \end{gather*}

\end{definition}

\begin{definition}
    The \emph{non-standard reduction relation} $\nsred$ on the set of sequential programs is defined
    by the set of rules in Figure~\ref{fig:seq_reduction2} together with all the rules in
    Figure~\ref{fig:seq_reduction} (with $\sred$ replaced by $\nsred$), except for \rn{SR-Cho-L} and  \rn{SR-Cho-R}.
    To simplify the notation, we may write \( \Exp \nsred_\Def \Exp' \) for \( (\Def, \Exp) \nsred (\Def, \Exp') \) or even \( \Exp \nsred \Exp' \) if \( \Def \) is clear from the context.
\end{definition}

\begin{figure}[tb]
    \centering
    \small
    \begin{minipage}{\linewidth}
        \centering
        \begin{prooftree}
            \AxiomC{\( \Exp \expequiv \Exp_1 \quad (\Def, \Exp_1) \nsred (\Def, \Exp'_1) \quad \Exp'_1 \expequiv \Exp' \)}
            \RightLabel{\textsc{(SR-Cong)}}
            \UnaryInfC{$(\Def, \Exp) \nsred (\Def, \Exp')$}
        \end{prooftree}
    \end{minipage}
    \begin{minipage}{\linewidth}
        \centering
        \begin{prooftree}
            \AxiomC{$(\Def, \Exp_1) \nsred (\Def, \Exp_1')$}
            \RightLabel{\textsc{(SR-ChoBody-L)}}
            \UnaryInfC{$(\Def, \Exp_1 \nondet \Exp_2) \nsred (\Def, \Exp_1' \nondet \Exp_2)$}
        \end{prooftree}
    \end{minipage}
    \begin{minipage}{\linewidth}
        \centering
        \begin{prooftree}
            \AxiomC{$(\Def, \Exp_2) \nsred (\Def, \Exp_2')$}
            \RightLabel{\textsc{(SR-ChoBody-R)}}
            \UnaryInfC{$(\Def, \Exp_1 \nondet \Exp_2) \nsred (\Def, \Exp_1 \nondet \Exp_2')$}
        \end{prooftree}
    \end{minipage}

    \normalsize
    \caption{Additional rules for the  non-standard reduction relation}
    \label{fig:seq_reduction2}
\end{figure}

For the proof of the soundness of our transformation
(given in Appendix~\ref{sec:soundness}),
we also prepare a relation \( \Def \subdef \Def' \), which intuitively
means that \(\Def\) can simulate \(\Def'\) so that if \((\Def,\Exp)\) is terminating,
so is \((\Def',\Exp)\)
(cf.\ Lemma~\ref{lem:subdef}).

\begin{figure}[tb]
    \centering
    \small
    \begin{minipage}{.4\linewidth}
        \centering
        \begin{prooftree}
            \AxiomC{ }
            \RightLabel{\textsc{(D-Id)}}
            \UnaryInfC{\( \Def \subdef \Def\)}
        \end{prooftree}
    \end{minipage}
    \begin{minipage}{.4\linewidth}
        \centering
        \begin{prooftree}
            \AxiomC{\( \Def = \Def_1 \mrg \Def_2 \)}
            \RightLabel{\textsc{(D-Splt)}}
            \UnaryInfC{\( \Def \subdef \Def_1 \)}
        \end{prooftree}
    \end{minipage}
    \begin{minipage}{\linewidth}
        \centering
        \begin{prooftree}
            \AxiomC{\( \Def = (\ndlet{x}{\Def'}) \) \qquad \( \len{\seq{x}} = \len{\seq{v}} \)}
            \RightLabel{\textsc{(D-ND)}}
            \UnaryInfC{\( \Def \subdef [\seq{v}/\seq{x}]\Def' \)}
        \end{prooftree}
    \end{minipage}
    \begin{minipage}{.4\linewidth}
        \centering
        \begin{prooftree}
            \AxiomC{ \( \Def_1 \subdef \Def_1' \)}
            \RightLabel{\textsc{(D-Mrg)}}
            \UnaryInfC{\( \Def_1 \mrg \Def_2 \subdef \Def_1' \mrg \Def_2 \)}
        \end{prooftree}
    \end{minipage}
    \begin{minipage}{.4\linewidth}
        \centering
        \begin{prooftree}
            \AxiomC{ \( \Def_1 \subdef \Def_2 \) \qquad \( \Def_2 \subdef \Def_3 \)}
            \RightLabel{\textsc{(D-Trns)}}
            \UnaryInfC{\( \Def_1 \subdef \Def_3 \)}
        \end{prooftree}
    \end{minipage}
    \normalsize
    \caption{Preorder on function definitions}
    \label{fig:subdef}
\end{figure}

\begin{lemma}
    \label{lem:subdef}
    Suppose that \( \Def \subdef \Def' \) and \( (\Def', \Exp ) \nsred ( \Def', \Exp') \).
    Then \( (\Def, \Exp ) \nsred^+ ( \Def, \Exp') \).
\end{lemma}
\begin{proof}
    By induction on the derivation of \( \Def \subdef \Def' \).
    \qed
\end{proof}
\section{Proof of the Soundness}  \label{sec:soundness}
Here we prove the soundness of the translation (Theorem~\ref{thm:soundness}) saying that if the sequential program \( (\Def, \Exp) \) obtained by translating \( P \) is terminating, \( P \) is also terminating.
The proof is split into two steps.
First, we show that reductions from \( P \) can be simulated by non-standard reductions from \( (\Def, \Exp ) \) (Lemma~\ref{lem:simulate}).
This implies that if \( (\Def, \Exp) \) is terminating with respect to the non-standard reduction, then \( P \) is terminating.
Then we show that if  \( (\Def, \Exp) \) is terminating with respect to the standard reduction, then \( (\Def, \Exp) \) is terminating with respect to the non-standard reduction (Lemma~\ref{lem:konig}).

We start by preparing some auxiliary lemmas that are used to show the simulation relation.

\begin{lemma}[substitution]
    \label{lem:subst}
    If \( \env; \chenv \vdash \seq{v} : \seq{\ty} \),
       \( \env; \chenv \vdash \seq{w} : \seq{\chty} \)
       and \( \env, \seq{y}:\seq{\ty}; \chenv, \seq{z}:\seq{\chty} \vdash P \Rightarrow
       \prog{\Def}{\Exp} \),
       then \( \env; \chenv \vdash [\seq{v}/\seq{y}, \seq{w} / \seq{z}]P \Rightarrow
       \prog{[\seq{v}/\seq{y}]\Def}{[\seq{v}/\seq{y}] \Exp} \).
\end{lemma}
\begin{proof}
    By induction on the derivation of $\env, \seq{y}:\seq{\ty}; \chenv, \seq{z}:\seq{\chty} \vdash P \Rightarrow \prog{\Def}{\Exp}$.
    \qed
\end{proof}

\begin{lemma} \label{lem:cong}
    If $P \piequiv P'$ and $\env; \chenv \vdash P \Rightarrow \prog{\Def}{\Exp}$,
    then there exists $\Exp'$ such that
    $\Exp \expequiv \Exp'$
    and $\env; \chenv \vdash P' \Rightarrow \prog{\Def}{\Exp'}$.
\end{lemma}
\begin{proof}
    By induction on the construction of $P \piequiv P'$.
    \qed
\end{proof}

Now we prove the simulation relation.

\begin{lemma} \label{lem:simulate}
    If $P \red P'$ and $\env; \chenv \vdash P \Rightarrow \prog{\Def}{\Exp}$,
    then there exist $\Def'$, $\Exp'$ such that \( \Def \subdef \Def' \),
    \( (\Def', \Exp) \nsred^+ (\Def', \Exp') \)
    and $\env; \chenv \vdash P' \Rightarrow \prog{\Def'}{\Exp'}$.
\end{lemma}
\begin{proof}
    By induction on the construction of \( P \red P' \).
    We only give detailed proofs for  interesting cases; the other cases are sketched.

\begin{description}
    \item[Case \rn{R-Comm}:]
    In this case $P \to P'$ must be of the form
    \begin{align*}
    \inexp{x}{\seq{y}}{\seq{z}}P_1 \PAR \outexp{x}{\seq{v}}{\seq{w}}P_2 \red [\seq{i}/ \seq{y} ,\seq{w} / \seq{z}]P_1 \PAR P_2,
    \end{align*}
    where
    $\len{\seq{y}} = \len{\seq{v}}$,
    $\len{\seq{z}} = \len{\seq{w}}$
    and $\seq{v} \Downarrow \seq{i}$.
    Also $\env; \chenv \vdash P \Rightarrow \prog{\Def}{\Exp}$ must be the form of
    \begin{align*}
      \env; \chenv \vdash P \Rightarrow
      \prog{(\ndlet{y}{\Def_1}) \mrg \Def_2}{(\ndlet{y}{\Exp_1}) \nondet (f_\reg(\seq{v}) \oplus \Exp_2)},
    \end{align*}
    where
    \begin{align}
    &\env; \chenv \vdash x : \Chty{\reg}{\seq{\ty}}{\seq{\chty}}
    \qquad \env; \chenv \vdash \seq{v} : \seq{\ty}
    \qquad \env; \chenv \vdash \seq{w} : \seq{\chty} \nonumber \\
    &\env, \seq{y} : \seq{\ty}; \chenv, \seq{z} : \seq{\chty} \vdash P_1 \Rightarrow
    \prog{\Def_1}{\Exp_1} \label{eq:sim:com:trans-p1} \\
    & \env; \chenv \vdash P_2 \Rightarrow \prog{\Def_2}{\Exp_2}. \label{eq:sim:com:trans-p2}
    \end{align}
    By applying Lemma~\ref{lem:subst} to \eqref{eq:sim:com:trans-p1} with
    $\env; \chenv \vdash \seq{i} : \seq{\ty}$,
    $\env; \chenv \vdash \seq{w} : \seq{\chty}$,
    we obtain
    $\env; \chenv \vdash [\seq{i} / \seq{y}, \seq{w})/\seq{z}]P_1 \Rightarrow
    \prog{[\seq{i}/\seq{y}]\Def_1}{[\seq{i}/\seq{y}]\Exp_1}$.
    From this and \eqref{eq:sim:com:trans-p2}, we have
    \begin{align*}
      \env; \chenv \vdash [\seq{i}/ \seq{y}, \seq{w}/ \seq{z}]P_1 \PAR P_2 \Rightarrow
      \prog{[\seq{i}/\seq{y}]\Def_1 \mrg \Def_2}
           {[\seq{i}/\seq{y}]\Exp_1 \nondet \Exp_2}
    \end{align*}
    by applying the rule \rn{SX-Par}.
    Observe that we also have
    \begin{align*}
      \Def = (\ndlet{y}{\Def_1}) \mrg \Def_2 \subdef [\seq{i}/\seq{y}]\Def_1 \mrg \Def_2.
    \end{align*}
    Therefore, for \( (\Def', \Exp')\) we can take \( ([\seq{i}/\seq{y}]\Def_1 \mrg \Def_2, [\seq{i}/\seq{y}]\Exp_1 \nondet \Exp_2 )\) with the following matching reduction sequence:
    \begin{align*}
      \Exp
      &=  (\ndlet{y}{\Exp_1}) \nondet f_\reg(\seq{v}) \nondet \Exp_2 \\
      &\nsred_{\Def'} [\seq{i}/\seq{y}]\Exp_1 \nondet f_\reg(\tilde{v}) \oplus E_2 \tag{\rn{SR-LetND}} \\
        &\nsred_{\Def'}, [\seq{i}/\seq{y}]\Exp_1 \nondet \skipexp \nondet \Exp_2 \tag{by (\rn{SR-App}) and \( \lambda \seq{y}.\skipexp \in \Def'(f_\reg) \) } \\
        &\expequiv [\seq{i}/\seq{y}]\Exp_1 \nondet \Exp_2.
    \end{align*}

    \item[Case \rn{R-RComm}:]
    In this case  $P \to P'$ is of the form
    \begin{align*}
      \rinexp{x}{\seq{y}}{\seq{z}}P_1 \PAR \outexp{x}{\seq{v}}{\seq{w}}P_2 \red \rinexp{x}{\seq{y}}{\seq{z}}P_1 \PAR [(\seq{i},\seq{w})/(\seq{y},\seq{z})]P_1 \PAR P_2,
    \end{align*}
    where
    $\len{\seq{y}} = \len{\seq{v}}$,
    $\len{\seq{z}} = \len{\seq{w}}$
    and $\seq{v} \Downarrow \seq{i}$.
    Moreover, the judgment $\env; \chenv \vdash P \Rightarrow
    \prog{\Def}{\Exp}$ must be of the form
    \begin{align*}
      \env; \chenv \vdash P \Rightarrow
      \prog{\{ \fdef{f_\reg}{\seq{y}}{\Exp_1} \} \mrg (\ndlet{y}{\Def_1}) \mrg \Def_2}
      {\skipexp \nondet f_\reg (\seq{v}) \nondet \Exp_2},
    \end{align*}
    where
    \begin{align}
      &\env; \chenv \vdash x : \Chty{\reg}{\seq{\ty}}{\seq{\chty}}
      \qquad \env; \chenv \vdash \seq{v} : \seq{\ty}
      \qquad \env; \chenv \vdash \seq{w} : \seq{\chty} \nonumber \\
      &\env; \chenv \vdash \rinexp{x}{\seq{y}}{\seq{z}}P_1 \Rightarrow
      \prog{\{ \fdef{f_\reg}{\seq{y}}{\Exp_1} \} \mrg (\ndlet{y}{\Def_1})}{
        \skipexp} \label{eq:sim:rcom:trans-in} \\
      &\env, \seq{y} : \seq{\ty}; \chenv, \seq{z} : \seq{\chty} \vdash P_1 \Rightarrow \Def_1; \Exp_1 \label{eq:sim:rcom:trans-p1}\\
      &\env; \chenv \vdash P_2 \Rightarrow \prog{\Def_2}{\Exp_2}. \label{eq:sim:rcom:trans-p2}
    \end{align}
    Since
    $\env; \chenv \vdash \seq{i} : \tilde{\ty}$ and
    $\env; \chenv \vdash \tilde{w} : \tilde{\chty}$, we can apply the substitution lemma (Lemma~\ref{lem:subst}) to \eqref{eq:sim:rcom:trans-p1} and obtain
    \begin{align*}
      \env; \chenv \vdash [\seq{i} / \seq{y}, \seq{w}/\seq{z}]P_1 \Rightarrow
          \prog{[\seq{i}/\seq{y}]\Def_1}{[\seq{i}/\seq{y}]\Exp_1}.
    \end{align*}
    From this, \eqref{eq:sim:rcom:trans-in} and \eqref{eq:sim:rcom:trans-p2}, we have
    \begin{align*}
      \env; \chenv \vdash P' \Rightarrow
      \begin{aligned}
        &(\{ \fdef{f_\reg}{\seq{y}}{\Exp_1} \} \mrg (\ndlet{y}{\Def_1}) \mrg  [\seq{i}/\seq{y}]\Def_1 \mrg \Def_2, \\
        &\skipexp \nondet [\seq{i}/\seq{y}]\Exp_1 \nondet \Exp_2)
      \end{aligned}
    \end{align*}
    So we can take \( \{ \fdef{f_\reg}{\seq{y}}{\Exp_1} \} \mrg (\ndlet{y}{\Def_1}) \mrg  [\seq{i}/\seq{y}]\Def_1 \mrg \Def_2 \) as  \( \Def' \) and \( \skipexp \nondet [\seq{i}/\seq{y}]\Exp_1 \nondet \Exp_2 \) as \( \Exp' \).
    Now it remains to show that \( \Def \subdef \Def' \) and that there is a reduction sequence from \( (\Def', \Exp) \) to \( (\Def', \Exp') \).
    The relation \( \Def \subdef \Def' \) holds because
    \begin{align*}
      \Def
      &= (\{ \fdef{f_\reg}{\seq{y}}{\Exp_1} \} \mrg (\ndlet{y}{\Def_1}) \mrg \Def_2 \\
      &= \{ \fdef{f_\reg}{\seq{y}}{\Exp_1} \} \mrg (\ndlet{y}{\Def_1}) \mrg (\ndlet{y}{\Def_1}) \mrg \Def_2 \\
      &\subdef \{ \fdef{f_\reg}{\seq{y}}{\Exp_1} \} \mrg (\ndlet{y}{\Def_1}) \mrg [\seq{i} / \seq{y}]{\Def_1} \mrg \Def_2 \tag{\rn{D-ND}} \\
      &=\Def'
    \end{align*}

    Finally, by \rn{SR-App}, we obtain
    \begin{align*}
      \Exp &= \skipexp \nondet f_\reg (\seq{v}) \nondet \Exp_2  \nsred_{\Def'} \skipexp \nondet [\seq{i}/\seq{y}]\Exp_1 \nondet \Exp_2  = \Exp'
    \end{align*}
    as desired.
    \item[Case \rn{R-If-T}:]
      In this case \( P \red P' \) and \( \env; \chenv \vdash P \Rightarrow
      \prog{\Def}{\Exp} \) must be of the form
      \begin{align*}
      &\ifexp{v}{P_1}{P_2} \red P_1 \\
        &\env; \chenv \vdash \ifexp{v} {P_1}{P_2} \Rightarrow
        \prog{\Def_1 \mrg \Def_2}{\ifexp{v}{\Exp_1}{\Exp_2}}
      \end{align*}
      where
      \begin{align*}
        v \Downarrow i \neq 0  \qquad \env; \chenv \vdash v : \ty \\
        \env; \chenv \vdash P_1 \Rightarrow \prog{\Def_1}{\Exp_1} \\
        \env; \chenv \vdash P_2 \Rightarrow \prog{\Def_2}{\Exp_2}.
      \end{align*}
    We can take \( (\Def_1, \Exp_1) \) for \( ( \Def', \Exp' )\) because \( \Def_1 \mrg \Def_2 \subdef \Def_1 \), and \( \Exp \nsred_{\Def_1} \Exp_1 \), which is trivial from \rn{SR-If-T}.

    \item[Case \rn{R-If-F}:]
    Similar to the previous case.

    \item[Case \rn{R-Cong}:]
    In this case \( P \red P' \) must be of the form
    \begin{align*}
      P \piequiv P_1 \red P_1' \piequiv P'.
    \end{align*}
    By Lemma~\ref{lem:cong}, we have
    \begin{align*}
      \env, \chenv \vdash P_1 \Rightarrow \prog{\Def}{\Exp_1} \text{ and } \Exp \expequiv \Exp_1
    \end{align*}
    for some \( \Exp_1 \).
    Thus, by the induction hypothesis, we have
    \begin{align}
      &\env, \chenv \vdash P_1' \Rightarrow \prog{\Def'}{\Exp_1'} \label{eq:sim:cong-P1prime}\\
      & (\Def',  \Exp_1) \nsred^+ (\Def', \Exp_1') \label{eq:sim:cong:red-seq}
    \end{align}
    where \( \Def \subdef \Def' \).
    By applying Lemma~\ref{lem:cong} to \eqref{eq:sim:cong-P1prime}, we obtain
    \begin{align*}
      \env, \chenv \vdash P' \Rightarrow \prog{\Def'}{\Exp'} \text{ and } \Exp_1' \expequiv \Exp'
    \end{align*}
    for some \( \Exp' \).
    It remains to show that \( (\Def', \Exp) \nsred^+  (\Def', \Exp') \), but this is easily shown by repeatedly applying the rule \rn{SR-Cong} along the reduction sequence \eqref{eq:sim:cong:red-seq}.

    \item[Case \rn{R-Par}, \rn{R-Nu} and \rn{R-LetND}:]
    Similar to the previous case, i.e.~follows from the definition of the translation and the induction hypothesis together with Lemma~\ref{lem:subdef}.
\end{description}

\leavevmode\qed
\end{proof}

\begin{lemma}  \label{lem:infinitechain}
    Suppose that \( \emptyset; \emptyset \vdash P \Rightarrow \prog{\Def}{\Exp} \).
    If \( (\mathcal{D}, E) \) is terminating with respect to $\nsred$, then \( P \) is terminating.
\end{lemma}
\begin{proof}
    We show the contraposition.
    Assume that \( P \) is not terminating, i.e.~assume that there exists an infinite reduction sequence $P = P_0 \red P_1 \red \cdots$.
    Let \( \Def_0 = \Def \) and \( \Exp_0 = \Exp \).
    By applying Lemma~\ref{lem:simulate}, for each natural number \( k \ge 1 \), we obtain \( \Def_k \), \( \Exp_k \) such that \( \emptyset; \emptyset \vdash P_k \Rightarrow
    \prog{\Def_k}{\Exp_k} \), \( (\Def_{k}, \Exp_{k-1}) \nsred^+ (\Def_{k}, \Exp_{k}) \) and \( \Def \subdef \Def_k \).
    Hence, by Lemma~\ref{lem:subdef} there exists an infinite reduction sequence
    $(\Def, \Exp) = (\Def, \Exp_0)\allowbreak \nsred^+ (\Def, \Exp_1) \nsred^+ \cdots$.
    \qed
\end{proof}

We now show the relation between standard and non-standard reductions.
\begin{lemma}  \label{lem:konig}
    Assume that $\emptyset; \emptyset \vdash P \Rightarrow \prog{\Def}{\Exp}$.
    If \( (\Def, \Exp) \) is terminating with respect to the standard reduction \( \sred \), then \( (\Def, \Exp)\) is also terminating with respect to the non-standard reduction relation \( \nsred \).
\end{lemma}

To prove the lemma above, we introduce a slight variation of the
non-standard reduction relation:
\((\Def,\Exp) \nsredv{\gamma} (\Def',\Exp')\) where
\(\gamma\in \set{1,2}^*\). (Actually, \(\Def\) does not change during the reduction.)
It is defined by the rules in Figure~\ref{fig:nsredv}.

\begin{figure}[tbp]

  \infrule[NSR-LetND]
          {\len{\seq{x}} = \len{\seq{i}}}
          {(\Def, \ndlet{x}{\Exp} )  \nsredv{\epsilon} (\Def, [\seq{i}/\seq{x}]\Exp)}
\vspace*{1ex}
          \infrule[NSR-App]
{ (\lambda \seq{y}.\Exp) \in  \Def(f)\andalso
 \len{\seq{y}} = \len{\seq{v}}\andalso
 \seq{v} \Downarrow \seq{i}}
{(\Def, f(\seq{v})) \nsredv{\epsilon} (\Def, [\seq{i}/\tilde{y}] \Exp)}

\vspace*{1ex}
\infrule[NSR-If-T]{v \Downarrow i \andalso i \neq 0}
 {(\Def, \ifexp{v}{\Exp_1}{\Exp_2})  \nsredv{\epsilon} (\Def, \Exp_1)}
\vspace*{1ex}
\infrule[NSR-If-F]{v \Downarrow 0}
 {(\Def, \ifexp{v}{\Exp_1}{\Exp_2})  \nsredv{\epsilon} (\Def, \Exp_2)}
\vspace*{1ex}

 \infrule[NSR-ChoBody-L]
 {(\Def, \Exp_1) \nsredv{\gamma} (\Def, \Exp_1')}
 {(\Def, \Exp_1 \nondet \Exp_2) \nsredv{1\cdot \gamma} (\Def, \Exp_1' \nondet \Exp_2)}
\vspace*{1ex}
 \infrule[NSR-ChoBody-R]
 {(\Def, \Exp_2) \nsredv{\gamma} (\Def, \Exp_2')}
 {(\Def, \Exp_1 \nondet \Exp_2) \nsredv{2\cdot \gamma} (\Def, \Exp_1 \nondet \Exp_2')}

\vspace*{1ex}
 \infrule[NSR-Ass-T]
  {v \Downarrow i\andalso i \neq 0}
  {(\Def, \textbf{Assume}(v);E) \nsredv{\epsilon} (\Def, \Exp)}
\vspace*{1ex}
 \infrule[NSR-Ass-F]
  {v \Downarrow 0}
  {(\Def, \textbf{Assume}(v);E) \nsredv{\epsilon} (\Def, \skipexp)}

\caption{A variation of the non-standard reduction relation}
\label{fig:nsredv}
\end{figure}

The only differences of
\((\Def,\Exp) \nsredv{\gamma} (\Def',\Exp')\) from
\((\Def,\Exp) \nsred (\Def',\Exp')\) are that
the reduction is annotated with the position \(\gamma\) that indicates where the reduction occurs,
and that the rule \rn{SR-Cong} for shuffling expressions is forbidden.
Since the rule \rn{SR-Cong} does not affect the reducibility, we can easily
observe the following property. (We omit the proof since it is trivial.)
\begin{lemma}
  \label{lem:nsred-vs-nsredv}
  If \((\Def,\Exp)\) has an infinite reduction sequence with respect to \(\nsred\),
  \((\Def,\Exp)\) has an infinite reduction sequence also with respect to \(\nsredv{\gamma}\).
\end{lemma}

It remains to show that
if \((\Def,\Exp)\) has an infinite reduction sequence
\[(\Def,\Exp)\nsredv{\gamma_1} (\Def,\Exp_1)\nsredv{\gamma_2}
(\Def,\Exp_2)\nsredv{\gamma_3}(\Def,\Exp_3)\nsredv{\gamma_4}\cdots,\]
then
\((\Def,\Exp)\) has an infinite reduction sequence also with respect to \(\sred\).

We write \(\gamma \preceq \gamma'\) if \(\gamma\) is a prefix of \(\gamma'\).
We have the following property.
\begin{lemma}
  \label{lem:inf-nsredv}
  If
\[(\Def,\Exp)\nsredv{\gamma_1} (\Def,\Exp_1)\nsredv{\gamma_2}
(\Def,\Exp_2)\nsredv{\gamma_3}(\Def,\Exp_3)\nsredv{\gamma_4}\cdots,\]
then there exists an infinite sequence
\(i_1 < i_2 < i_3< \cdots\)
such that \(\gamma_{i_j} \preceq \gamma_{i_k}\) for any \(j<k\).
\end{lemma}
\begin{proof}
  The required property obviously holds if   the set \(\set{\gamma_i\mid i\ge 1}\) is finite.
  So, assume that \(\set{\gamma_i\mid i\ge 1}\) is infinite.
  Let \(T\) be the least binary tree that contains, for every \(\gamma_i\),
  the node whose path from the root is \(\gamma_i\).
  By the assumption that \(\set{\gamma_i\mid i\ge 1}\) is infinite,
  \(T\) is an infinite tree. Thus, by K\"onig's lemma,
  \(T\) must have an infinite path, which implies that
  there exists an infinite sequence
  \[\gamma_{i_1} \preceq \gamma_{i_2} \preceq \gamma_{i_3} \preceq \cdots, \]
  as required. \qed
\end{proof}

For an expression \(\Exp\) and a position \(\gamma\in\set{1,2}^*\), we write
\(\Proj{\Exp}{\gamma}\) for the subexpression at \(\gamma\). It is inductively defined by:
\[
\begin{array}{l}
\Proj{\Exp}{\epsilon} = \Exp\\
\Proj{\Exp}{i\cdot \gamma} =
\left\{\begin{array}{ll}
  \Proj{\Exp_i}{\gamma} & \mbox{if $\Exp$ is of the form \(\Exp_1\nondet \Exp_2\)}\\
  \mbox{undefined}\hspace*{2em} & \mbox{otherwise}
\end{array}\right.
\end{array}
\]
The following lemma states the correspondence between \(\nsredv{\gamma}\) and \(\sred\).

\begin{lemma}
\label{lem:nsredv-vs-sred}
  \begin{enumerate}
\item  If \((\Def,\Exp)\nsredv{\gamma}(\Def,\Exp')\),
  then \((\Def, \Proj{\Exp}{\gamma})\sred (\Def,\Proj{\Exp'}{\gamma})\).
\item Suppose \(\Proj{\Exp}{\gamma'}\) is defined and \(\gamma'\not\preceq \gamma\).
  If \((\Def,\Exp)\nsredv{\gamma}(\Def,\Exp')\), then
  \(\Proj{\Exp}{\gamma'}=\Proj{\Exp'}{\gamma'}\).
\item
  If \((\Def,\Exp)\nsredv{\gamma}(\Def,\Exp')\), and \(\gamma'\preceq \gamma\),
  then \((\Def, \Proj{\Exp}{\gamma'} ) \sred^* (\Def, \Proj{\Exp}{\gamma})\).
\end{enumerate}  
\end{lemma}
\begin{proof}
  The properties follow by a straightforward induction on the derivation of
  \((\Def,\Exp)\nsredv{\gamma}(\Def,\Exp')\). \qed
\end{proof}

We are now ready to prove Lemma~\ref{lem:konig}.

\begin{proof}[of Lemma~\ref{lem:konig}]
  We show the contraposition.
  Suppose \((\Def,\Exp)\) has an infinite reduction sequence with respect to \(\nsred\).
  By Lemma~\ref{lem:nsred-vs-nsredv}, there exists an infinite reduction sequence
  \[(\Def,\Exp)\nsredv{\gamma_1} (\Def,\Exp_1)\nsredv{\gamma_2}
(\Def,\Exp_2)\nsredv{\gamma_3}(\Def,\Exp_3)\nsredv{\gamma_4}\cdots.\]
  By Lemma~\ref{lem:inf-nsredv},
 there exists an infinite sequence:
  \[\gamma_{i_1}\preceq \gamma_{i_2}\preceq \gamma_{i_3}\preceq \cdots.\]
  such that \(i_1<i_2<i_3<\cdots\).
  Let us choose a maximal one among such sequences, i.e.,
  a sequence
  \[\gamma_{i_1}\preceq \gamma_{i_2}\preceq \gamma_{i_3}\preceq \cdots.\]
  such that, for any \(i_j\), 
  \(\gamma_{k}\preceq \gamma_{i_{j}}\) implies \(k=i_{j'}\) for some \(j'\le j\).
  Consider the fragment of the infinite reduction sequence:
  \[
  (\Def,\Exp_{i_{\ell-1}})\nsredv{\gamma_{i_{\ell-1}+1}} (\Def,\Exp_{i_{\ell-1}+1})
  \nsredv{\gamma_{i_{\ell-1}+2}}\cdots
  \nsredv{\gamma_{i_{\ell}-1}}
(\Def,\Exp_{i_{\ell}-1})\nsredv{\gamma_{i_\ell}}(\Def,\Exp_{i_\ell})
  \]
  for each \(\ell>0\). 
  (Here, we define \(\gamma_0 = \epsilon\), \(i_0=0\) and \(E_0 = E\).)
  By Lemma~\ref{lem:nsredv-vs-sred} (1)
  and \((\Def,\Exp_{i_{\ell}-1})\nsredv{\gamma_{i_\ell}}(\Def,\Exp_{i_\ell})\),
  we have
  \[(\Def,\Proj{\Exp_{i_{\ell}-1}}{\gamma_{i_\ell}})\sred
  (\Def,\Proj{\Exp_{i_{\ell}}}{\gamma_{i_\ell}}).\]
  By Lemma~\ref{lem:nsredv-vs-sred} (2) (note that
  since none of \(\gamma_{i_{\ell-1}+1},\ldots,\gamma_{i_{\ell}-1}\) is a
  prefix of \(\gamma_{i_{\ell}}\) by the assumption on maximality,
  \(\Proj{\Exp_{i_{\ell-1}}}{\gamma_{i_\ell}}\) is defined),
  we have
  \[\Proj{\Exp_{i_{\ell-1}}}{\gamma_{i_\ell}} =
  \Proj{\Exp_{i_{\ell-1}+1}}{\gamma_{i_\ell}} = \cdots =
  \Proj{\Exp_{i_{\ell}-1}}{\gamma_{i_\ell}}.\]
  Thus, together with Lemma~\ref{lem:nsredv-vs-sred} (3), we obtain:
  \[(\Def,\Proj{\Exp_{i_{\ell-1}}}{\gamma_{i_{\ell-1}}})\sred^*
  (\Def,\Proj{\Exp_{i_{\ell-1}}}{\gamma_{i_{\ell}}})\sred
  (\Def,\Proj{\Exp_{i_{\ell}}}{\gamma_{i_\ell}}).\]
  Therefore,
  we have an infinite reduction sequence
  \[(\Def,\Exp)=(\Def,\Proj{\Exp_{i_0}}{\gamma_{i_0}})\sred^+ (\Def,\Proj{\Exp_{i_1}}{\gamma_{i_1}})
  \sred^+ (\Def,\Proj{\Exp_{i_2}}{\gamma_{i_2}})
  \sred^+ (\Def,\Proj{\Exp_{i_3}}{\gamma_{i_3}})
  \sred^+ \cdots,\]
  as required. \qed
\end{proof}

Finally, the soundness (Theorem~\ref{thm:soundness}) follows from Lemmas~\ref{lem:infinitechain} and \ref{lem:konig}.

\section{Complete Definition of the Refinement Type System}
\label{sec:refinement-apx}
This section shows the complete definition of the refinement type system we discussed in Section~\ref{sec:refinement}.

First, we define the well-formedness conditions for types and type environments.
We write \(\FV(\pred)\) (\(\FV(\predenv)\), resp.)
for the set of variables occurring in \(\pred\) (\(\predenv\), resp.),
and \(\dom(\env)\) for the domain of \(\env\), i.e.,
\(\set{x \mid x\COL\ty\in\env}\).
The relations \(\tyok{\env}{\chty}\)
and \(\tyenvok{\env}{\predenv}{\chenv}\)
are defined by:

\infrule{\FV(\pred)\subseteq \dom(\env)\cup \set{\seq{x}}\\
  \tyok{\env, \seq{x}\COL\seq{\ty}}{\chty_i}\mbox{ for each $i\in\set{1,\ldots,k}$}}
{\tyok{\env}{\rch{\reg}{\seq{x}}{\pred}{\chty_1,\ldots,\chty_k}}}

\infrule{\tyok{\env}{\chty} \mbox{ for every $x\COL\chty\in\chenv$}\\
   \FV(\predenv) \subseteq \dom(\env)}
{\tyenvok{\env}{\predenv}{\chenv}}
  For example, \(x\COL \ty \p \rch{\reg}{y}{y<x}{\epsilon}:\OK\)
  holds but
  \(\emptyset \p \rch{\reg}{y}{y<x}{\epsilon}:\OK\) does not.

For every type judgment of the form \(\env;\predenv;\chenv \p P\),
we implicitly require that \(\tyenvok{\env}{\predenv}{\chenv}\) holds.
Similarly,  for $\env;\predenv;\chenv \p v\COL\chty$,
we require
that \(\tyenvok{\env}{\predenv}{\chenv}\) and \(\tyok{\env}{\chty}\) hold.

The complete list of typing rules is given in Figure~\ref{fig:refinement_type_system-complete}.

\begin{figure}[tb]
    \centering
    \small
    \begin{minipage}{0.4\linewidth}
        \centering
        \begin{prooftree}
            \AxiomC{}
            \RightLabel{\textsc{(RT-Nil)}}
            \UnaryInfC{$\env;\predenv;\chenv\p \textbf{0}$}
        \end{prooftree}
    \end{minipage}
    \begin{minipage}{0.55\linewidth}
        \centering
        \begin{prooftree}
            \AxiomC{$\env;\predenv;\chenv\p P_1$}
            \AxiomC{$\env;\predenv;\chenv\p P_2$}
            \RightLabel{\textsc{(RT-Par)}}
            \BinaryInfC{$\env;\predenv;\chenv\p P_1 \PAR P_2$}
        \end{prooftree}
    \end{minipage}
    \\
    \begin{minipage}{\linewidth}
        \centering
        \begin{prooftree}
            \AxiomC{$\env;\predenv;\chenv\p x\COL\rch{\reg}{\seq{y}}{\pred}{\seq{\chty}}$}
            \AxiomC{$\env,\seq{y}\COL\seq{\ty}; \predenv,\pred; \chenv,\seq{z}\COL\seq{\chty} \p P$}
            \RightLabel{\textsc{(RT-In)}}
            \BinaryInfC{$\env;\predenv;\chenv\p \inexp{x}{\seq{y}}{\seq{z}}P$}
        \end{prooftree}
    \end{minipage}
    \\
    \begin{minipage}{0.93\linewidth}
        \centering
        \infrule[RT-Out]
        {\env;\predenv;\chenv\p x\COL\rch{\reg}{\seq{y}}{\pred}{\seq{\chty}}\andalso
        \env;\predenv;\chenv\p \seq{v}\COL\seq{\ty}\andalso
        \predenv \vDash [\seq{v}/\seq{y}]\pred\\
        \env;\predenv;\chenv\p \seq{w}\COL[\seq{v}/\seq{y}]\seq{\chty}\andalso
        \env;\predenv;\chenv\p P}
        {\env;\predenv;\chenv \vdash \outexp{x}{\seq{v}}{\seq{w}}P}
    \end{minipage}
    \\
    \begin{minipage}{\linewidth}
        \centering
        \begin{prooftree}
            \AxiomC{$\env;\predenv;\chenv,x\COL\chty \p P$}
            \RightLabel{\textsc{(RT-Nu)}}
            \UnaryInfC{$\env;\predenv;\chenv \p (\nu x\COL\chty)P$}
        \end{prooftree}
    \end{minipage}
    \\
    \begin{minipage}{\linewidth}
        \centering
        \begin{prooftree}
            \AxiomC{$\env;\predenv;\chenv\p x\COL\rch{\reg}{\seq{y}}{\pred}{\seq{\chty}}$}
            \AxiomC{$\env,\seq{y}\COL\seq{\ty}; \predenv,\pred; \chenv,\seq{z}\COL\seq{\chty} \p P$}
            \RightLabel{\textsc{(RT-RIn)}}
            \BinaryInfC{$\env;\predenv;\chenv\p \rinexp{x}{\seq{y}}{\seq{z}}P$}
        \end{prooftree}
    \end{minipage}
    \\
    \begin{minipage}{\linewidth}
        \centering
        \begin{prooftree}
            \AxiomC{$\env; \predenv; \chenv \p v\COL\ty$}
            \AxiomC{$\env; \predenv, v \neq 0; \chenv \p P_1$}
            \AxiomC{$\env; \predenv, v =    0; \chenv \p P_2$}
            \RightLabel{\textsc{(RT-If)}}
            \TrinaryInfC{$\env;\predenv;\chenv\p \ifexp{v}{P_1}{P_2}$}
        \end{prooftree}
    \end{minipage}
    \\
    \begin{minipage}{\linewidth}
        \begin{prooftree}
            \AxiomC{$\env,\seq{x}\COL\seq{\ty}; \predenv; \chenv \vdash P$}
            \RightLabel{\textsc{(RT-LetND)}}
            \UnaryInfC{$\env; \predenv; \chenv \vdash \ndlet{x}{P}$}
        \end{prooftree}
    \end{minipage}
    \\
    \begin{minipage}{0.45\linewidth}
        \centering
        \begin{prooftree}
            \AxiomC{$x \COL \ty \in \env$}
            \RightLabel{\textsc{(RT-Var-Int)}}
            \UnaryInfC{$\env;\predenv;\chenv \p x\COL\ty$}
        \end{prooftree}
    \end{minipage}
    \begin{minipage}{0.45\linewidth}
        \centering
        \begin{prooftree}
            \AxiomC{$x \COL \chty \in \chenv$}
            \RightLabel{\textsc{(RT-Var-Ch)}}
            \UnaryInfC{$\env;\predenv;\chenv \p x\COL\chty$}
        \end{prooftree}
    \end{minipage}
    \\
    \begin{minipage}{0.45\linewidth}
        \centering
        \begin{prooftree}
            \AxiomC{}
            \RightLabel{\textsc{(RT-Int)}}
            \UnaryInfC{$\env;\predenv;\chenv \p i\COL\ty$}
        \end{prooftree}
    \end{minipage}
    \begin{minipage}{0.45\linewidth}
        \centering
        \begin{prooftree}
            \AxiomC{$\env;\predenv;\chenv \p \seq{v}\COL\seq{\ty}$}
            \RightLabel{\textsc{(RT-Op)}}
            \UnaryInfC{$\env;\predenv;\chenv \p \op(\seq{v})\COL\ty$}
        \end{prooftree}
    \end{minipage}
    \normalsize
    \caption{Typing rules of the refinement type system for the $\pi$-calculus}
    \label{fig:refinement_type_system-complete}
\end{figure}

\section{Refinement Type System with Subtyping}
\label{sec:subtyping}

As mentioned in Section~\ref{sec:implementation},
the implementation is based on the following extension of the refinement type system in Section~\ref{sec:rtype}.

The set of \emph{refinement i/o channel types}, ranged over by $\chty$, is given by:
\[ 
    \chty ::= \ioch{\reg}{\seq{x}}{\inty{\pred}}{\inty{\seq{\chty}}}{\outy{\pred}}{\outy{\seq{\chty}}}
\]
Here, \(\ioch{\reg}{\seq{x}}{\inty{\pred}}{\inty{\seq{\chty}}}{\outy{\pred}}{\outy{\seq{\chty}}}\)
is the type of channels used for \emph{receiving} tuples \((\seq{x};\seq{y})\)
such that \(\seq{x}\) satisfies \(\pred_I\) and \(\seq{y}\) have types
\(\inty{\seq{\chty}}\),
and for \emph{sending} tuples \((\seq{x};\seq{y})\)
such that \(\seq{x}\) satisfies \(\pred_O\) and \(\seq{y}\) have types
\(\outy{\seq{\chty}}\).
The distinction between the types of input (i.e. received) values  and those of
output (i.e. sent) values has been inspired by
the type system of Yoshida and Hennessy~\cite{DBLP:conf/concur/YoshidaH99}.
It leads to a more precise type system than Pierce and Sangiorgi's subtyping,
and is convenient for automatic refinement type inference~\cite{Pierce96MSCS}
(because we need not infer input/output modes and perform case analysis on the modes).

The subtyping relation on channel types is defined by:
\infrule[RT-Sub-Ch]
{\predenv,\inty{\pred} \vDash \inty{\pred}' \andalso
 \env,\seq{x}\COL\seq{\ty};\predenv,\inty{\pred}\p  \inty{\seq{\chty}} \subtype \inty{\seq{\chty}}' \\
 \predenv,\outy{\pred}' \vDash \outy{\pred} \andalso
 \env,\seq{x}\COL\seq{\ty};\predenv,\outy{\pred}'\p  \outy{\seq{\chty}}' \subtype \outy{\seq{\chty}}
}
{\env;\predenv\p 
    \ioch{\reg}{\seq{x}}{\inty{\pred}}{\inty{\seq{\chty}}}{\outy{\pred}}{\outy{\seq{\chty}}} 
    \subtype 
    \ioch{\reg}{\seq{x}}{\inty{\pred}'}{\inty{\seq{\chty}}'}{\outy{\pred}'}{\outy{\seq{\chty}}'}
}
Note that the 
 channel type $\ioch{\reg}{\seq{x}}{\inty{\pred}}{\inty{\seq{\chty}}}{\outy{\pred}}{\outy{\seq{\chty}}}$
is covariant on $\inty{\pred}$ and $\inty{\seq{\chty}}$, 
and contravariant on $\outy{\pred}$ and $\outy{\seq{\chty}}$.

We make the following two modifications to the typing rules.
\begin{enumerate}
\item
  We add the following subsumption rule.
        \infrule[RT-Sub]
                {\env;\predenv;\chenv\p v\COL \chty\andalso \env;\predenv\p \chty \subtype \chty'}
                {\env;\predenv;\chenv\p v\COL \chty'}

  \item
    We refine the well-formedness condition on types and type environments by:
\infrule{\FV(\pred)\subseteq \dom(\env)\cup \set{\seq{x}}\\
  \tyok{\env, \seq{x}\COL\seq{\ty}}{\inty{\seq{\chty}}}\\
  \tyok{\env, \seq{x}\COL\seq{\ty}}{\outy{\seq{\chty}}}
}
{\tyok{\env}{\ioch{\reg}{\seq{x}}{\inty{\pred}}{\inty{\seq{\chty}}}{\outy{\pred}}{\outy{\seq{\chty}}}}}

\infrule{
   \FV(\predenv) \subseteq \dom(\env)
}
{\tyenvok{\env}{\predenv}{\emptyset}}

\infrule{\tyok{\env}{\ioch{\reg}{\seq{x}}{\inty{\pred}}{\inty{\seq{\chty}}}{\outy{\pred}}{\outy{\seq{\chty}}}}\\
  \predenv, \outy{\pred} \models \inty{\pred}
  \andalso
  \env,\seq{x}\COL\seq{\ty};\predenv,\outy{\pred}\p \outy{\seq{\chty}} \subtype
  \inty{\seq{\chty}}}
{\tyenvok{\env}{\predenv}{\chenv,y\COL\ioch{\reg}{\seq{x}}{\inty{\pred}}{\inty{\seq{\chty}}}{\outy{\pred}}{\outy{\seq{\chty}}}}}

The requirement for the subtyping relation in the last rule
 ensures the consistency between the types of
values expected by a receiver process and those actually output by a sender process;
for example, the channel type
\(\ioch{\reg}{x}{x>0}{\epsilon}{x<0}{\epsilon}\) is judged to be ill-formed,
because the type indicates that a receiver process expects a positive value \(x\) but 
a sender will output a negative value on the channel.
\end{enumerate}

The following example demonstrates the usefulness of subtyping for refinement channel
types.
\begin{example}
  \label{ex:subtyping}
  Let us consider the following process:
  \[
  \begin{array}{l}
    \rinexp{\pre}{x}{r}\soutatom{r}{x-1}\\
\PAR
\srinexp{f}{x}\ifexp{x<0}{\zeroexp}
        {\nuexp{s}(\outatom{\pre}{x}{s}\PAR \sinexp{s}{y}\soutatom{f}{y})}\\
\PAR \soutatom{f}{100}\\
\PAR \rinexp{c}{x}{r}\letexp{y}{\ndint}\soutatom{r}{y}\\
\PAR \soutatom{d}{\pre}\PAR \soutatom{d}{c}
  \end{array}
  \]
  The process consisting of the first three lines
  is a variation of the process in Example~\ref{ex:weakeness-of-basic-transformation},
  which is obviously terminating.
  Without the fourth and fifth lines, we would be able to assign
  the type \(\rch{\reg_1}{x}{\TRUE}{\rch{\reg_2}{y}{y<x}{\epsilon}}\)
  in the refinement type system in Section~\ref{sec:refinement},
  and reduce the process to a terminating program.

  The processes on the \skchanged{fifth} line, however, force us to assign the same type
  to \(\pre\) and \(c\) in the refinement type system in Section~\ref{sec:refinement},
  and thus we can assign only \(\rch{\reg_1}{x}{\TRUE}{\rch{\reg_2}{y}{\TRUE}{\epsilon}}\)
  to \(\pre\), failing to transform the process to a non-terminating program.

  With subtyping, we can assign the following types to \(\pre\), \(c\), and \(d\):
  \[
  \begin{array}{l}
    \pre\COL
    \ioch{\reg_1}{x}{\TRUE}{\ioch{\reg_2}{y}{\TRUE}{\epsilon}{y<x}{\epsilon}}{\TRUE}{\ioch{\reg_2}{y}{\TRUE}{\epsilon}{y<x}{\epsilon}}\\
    c\COL
    \ioch{\reg_1}{x}{\TRUE}{\ioch{\reg_2}{y}{\TRUE}{\epsilon}{\TRUE}{\epsilon}}{\TRUE}{\ioch{\reg_2}{y}{\TRUE}{\epsilon}{\TRUE}{\epsilon}}\\
    d\COL
    \ioch{\reg_0}{\epsilon}{\TRUE}{\chty}{\TRUE}{\chty}\\
    \mbox{where}\\
    \chty =
    \ioch{\reg_1}{x}{\TRUE}{\ioch{\reg_2}{y}{\TRUE}{\epsilon}{y<x}{\epsilon}}{\TRUE}{\ioch{\reg_2}{y}{\TRUE}{\epsilon}{\TRUE}{\epsilon}}\\
  \end{array}
  \]
  Note that the types of \(\pre\) and \(c\) are
  subtypes of \(\chty\).
  Here, the type of \(\pre\) indicates that
  the value \(y\) sent along the second argument \(r\) should be smaller than
  the first argument \(x\).
  Thus, the process on the second line is translated to the following
  function definition:
  \[
 \begin{array}{l}
\Fname{\reg_f}(x)=\ifexp{x<0}{\skipexp}\\\qquad\qquad{
  (\Fname{\pre}(x)\nondet (\letexp{y}{\ndint}\assexp{y<x}\Fname{\reg_f}(y)))}
\end{array}
\]

\end{example}

}

\end{document}